\documentclass{SciPost}

\usepackage[scaled=1]{XCharter}
\usepackage[bitstream-charter, expert,cal=cmcal]{mathdesign}
\usepackage{upgreek}

\RequirePackage[varqu,varl]{inconsolata}
\DeclareFontFamily{U}{nxlmi}{}
\DeclareFontSubstitution{U}{nxlmi}{m}{it}
\DeclareFontShape{U}{nxlmi}{m}{it}{
    <-6.3>    nxlmi05
    <6.3-8.6> nxlmi07
    <8.6->    nxlmi0
}{}
\DeclareFontShape{U}{nxlmi}{b}{it}{
    <-6.3>    nxlbmi05
    <6.3-8.6> nxlbmi07
    <8.6->    nxlbmi0
}{}
\renewcommand{\partial}{{\text{\scalebox{1.14}{\usefont{U}{nxlmi}{m}{it}\symbol{64}}}\mspace{1mu}}}

\makeatletter
\g@addto@macro\bfseries{\boldmath}
\makeatother

\PassOptionsToPackage{dvipsnames}{xcolor}
\definecolor{maroon}{RGB}{167,74,74}
\definecolor{magreen}{RGB}{59,125,37}
\definecolor{mablue}{RGB}{59,136,195}
\definecolor{mayellow}{RGB}{242,147,24}
\colorlet{mabrown}{maroon!50!magreen}
\colorlet{maorange}{maroon!50!mayellow}
\colorlet{macyan}{magreen!50!mablue}

\colorlet{red}{maroon}
\colorlet{green}{magreen}
\colorlet{blue}{mablue}
\colorlet{yellow}{mayellow}
\colorlet{brown}{mabrown}
\colorlet{orange}{maorange}
\colorlet{cyan}{macyan}

\definecolor{rred}{RGB}{167,33,74}
\definecolor{bblue}{RGB}{29,136,255}
\definecolor{ppurple}{RGB}{113,84,165}
\definecolor{ppink}{RGB}{255,55,219}

\RequirePackage{hyperref}
\hypersetup{colorlinks,
    linkcolor=mablue,
    citecolor=magreen,
    urlcolor=maorange,
    anchorcolor=mablue,
    filecolor=mablue,
    menucolor=mablue}

\makeatletter
\renewcommand\section{\@startsection{section}{1}{\z@}%
  {-3.5ex \@plus -1.3ex \@minus -.7ex}%
  {2.3ex \@plus.4ex \@minus .4ex}%
  {\large\scshape\bfseries}}
\renewcommand\subsection{\@startsection{subsection}{2}{\z@}%
    {-2.3ex\@plus -1ex \@minus -.5ex}%
    {1.2ex \@plus .3ex \@minus .3ex}%
    {\normalsize\scshape\bfseries}}
\renewcommand\subsubsection{\@startsection{subsubsection}{3}{\z@}%
    {-2.3ex\@plus -1ex \@minus -.5ex}%
    {1ex \@plus .2ex \@minus .2ex}%
    {\normalsize\scshape\bfseries}}
\renewcommand\paragraph{\@startsection{paragraph}{4}{\z@}%
    {1.75ex \@plus1ex \@minus.2ex}%
    {-1em}%
    {\normalsize\bfseries}}
\renewcommand\subparagraph{\@startsection{subparagraph}{5}{\parindent}%
    {1.75ex \@plus1ex \@minus .2ex}%
    {-1em}%
    {\normalsize\bfseries}}
\makeatother

\usepackage{twemojis}
\newcommand{\st}{\texttwemoji{monkey face}}
\newcommand{\jf}{\texttwemoji{fish}}

\usepackage{comment} 

\usepackage{import}
\usepackage{enumitem}
\usepackage{csquotes}

\RequirePackage[thicklines]{cancel}

\newcommand{\nn}{\nonumber \\}

\renewcommand{\emptyset}{\varnothing}

\let\coloneqq\relax
\let\eqqcolon\relax
\makeatletter
\newcommand*{\coloneqq}{\mathrel{\rlap{\raisebox{0.3ex}{$\m@th\cdot$}}\raisebox{-0.3ex}{$\m@th\cdot$}}=}
\newcommand*{\eqqcolon}{\mathrel{=\llap{\raisebox{0.3ex}{$\m@th\cdot$}}\llap{\raisebox{-0.3ex}{$\m@th\cdot$}}}}
\makeatother

\newcommand{\demeqq}{\overset{!}{=}}

\usepackage{physics}
\newcommand{\parti}{\mathcal{Z}}

\newcommand{\pd}{\partial} 

\newcommand{\inv}[1]{{{#1}^{-1}}} 
\newcommand{\trans}[1]{{{#1}^{\top}}} 
 
\newcommand{\Half}{\frac{1}{2}}

\newcommand{\suchthat}{\;\middle|\;}

\renewcommand{\vec}{\boldsymbol}

\newcommand{\w}{\mathbin{\scalebox{0.8}{$\wedge$}}} 
\newcommand{\ve}{\mathbin{\scalebox{0.8}{$\vee$}}} 
\newcommand{\ex}[1]{\mathrm{e}^{#1}} 
\newcommand{\ii}{\mathrm{i}} 
\newcommand{\pdual}[1]{\mathrm{PD}[#1]} 

\let\C\undefined
\let\U\undefined
\newcommand{\R}{\mathbb{R}}
\newcommand{\C}{\mathbb{C}}

\newcommand{\Z}{\mathbb{Z}}

\renewcommand{\S}{\mathbb{S}}

\newcommand{\set}[1]{\qty{#1}}

\newcommand{\eqv}{\Longleftrightarrow}

\newcommand{\U}{\mathrm{U}}

\renewcommand{\H}{\mathrm{H}} 

\renewcommand{\t}[1]{{\text{#1}}}
\newcommand{\nspace}[1][1]{\kern-#1em}

\newcommand{\blob}{\bullet}

\renewcommand{\mapsto}{\longmapsto}

\newcommand{\xto}[1]{\overset{#1}{\to}}
\newcommand{\xmapsto}[1]{\overset{#1}{\mapsto}}
\renewcommand{\mod}{\ \text{mod}\ }
\renewcommand{\emptyset}{\varnothing}

\makeatletter
\renewcommand{\xmapsto}[2][]{\ext@arrow 0599{\mapstofill@}{#1}{#2}}
\def\mapstofill@{\arrowfill@{\mapstochar\relbar}\relbar\rightarrow}
\makeatother

\def \cD {\mathcal{D}}
\def \cE {\mathcal{E}}

\def \cG {\mathcal{G}}
\def \cH {\mathcal{H}}

\def \cO {\mathcal{O}}

\def \cS {\mathcal{S}}

\DeclareMathAlphabet{\mathfrak}{U}{jkpmia}{m}{n}

\def \fA {\mathfrak{A}}

\def \fU {\mathfrak{U}}

\def \fZ {\mathfrak{Z}}

\def \fq {\mathfrak{q}}

\def \fv {\mathfrak{v}}
\def \fw {\mathfrak{w}}

\def \bbB {\mathbb{B}}
\def \bbC {\mathbb{C}}

\def \bbG {\mathbb{G}}

\def \bbK {\mathbb{K}}
\def \bbL {\mathbb{L}}

\def \bb1 {{\mathbb{1}}}

\def \sfI {\mathsf{I}}
\def \sfJ {\mathsf{J}}

\def \sfc {\mathsf{c}}

\def \sfi {\mathsf{i}}
\def \sfj {\mathsf{j}}
\def \sfk {\mathsf{k}}

\newcommand{\BF}{\t{BF}}

\newcommand{\W}{\mathrm{W}}

\newcommand{\hW}{\hat{\mathrm{W}}}
\newcommand{\hV}{\hat{\mathrm{V}}}
\newcommand{\aent}{R} 
\newcommand{\coaent}{\co{\aent}} 
\newcommand{\inr}[1]{\overline{#1}} 
\newcommand{\lattice}{\Lambda}
\renewcommand{\b}{\mathrm{b}} 

\newcommand{\CC}{\mathrm{C}}
\newcommand{\ZZ}{\mathrm{Z}}
\newcommand{\dH}{\check{\H}}
\newcommand{\Ocl}{\Omega_{\text{cl}}}

\newcommand{\ib}{\iota_{\scriptscriptstyle \pd}}

\newcommand{\lapl}{\triangle}

\newcommand{\beq}{\begin{equation}}
        \newcommand{\eeq}{\end{equation}}
\newcommand{\pa}{\partial}
\newcommand{\mf}[1]{\mathfrak{#1}}
\newcommand{\spn}{\text{span}}
\newcommand{\mc}[1]{\mathcal{#1}}
\newcommand{\eom}{\t{\textsc{eom}}}

\usepackage{cancel}
\newcommand{\hin}{\mathbin{\tilde{\in}}}
\newcommand{\nothin}{\mathbin{\cancel{\tilde{\in}}}}

\newcommand{\co}[1]{{#1}^\mathsf{c}}
\newcommand{\cco}[1]{{#1}^\mathsf{cc}}
\newcommand{\fall}{{{}^\forall\!}}

\newcommand{\h}{\mathrm{h}} 
\newcommand{\cpos}{\mathbf{q}} 
\newcommand{\cmom}{\mathbf{p}} 

\usepackage{tikz-cd}

\RequirePackage{mdframed}
\usepackage{empheq}
\newcommand*\obox[1]{%
    \fcolorbox{white}{blue!17}{\hspace{1em}#1\hspace{1em}}}

\newenvironment{claim}{\par\begin{minipage}{\linewidth}\textbf{Claim:}}{\end{minipage}}
\newenvironment{proof}{\begin{mdframed}[linecolor=blue!17,  hidealllines=true,leftline=true, linewidth=.2em, leftmargin=+1cm]\noindent\ignorespaces \textbf{Proof:}}{\end{mdframed}}

\newcommand{\alge}[1]{\cS_{#1}}
\newcommand{\vne}{\cS_\t{vN}}
\newcommand{\kb}[2]{\ketbra{#1}{#2}}

\def\fdiffd{\mathrm{D}}
\DeclareDocumentCommand\fdifferential{ o g d() }{ 
    \IfNoValueTF{#2}{
        \IfNoValueTF{#3}
        {\fdiffd\IfNoValueTF{#1}{}{^{#1}}}
        {\mathinner{\fdiffd\IfNoValueTF{#1}{}{^{#1}}\argopen(#3\argclose)}}
    }
    {\mathinner{\fdiffd\IfNoValueTF{#1}{}{^{#1}}#2} \IfNoValueTF{#3}{}{(#3)}}
}
\DeclareDocumentCommand\DD{}{\fdifferential}

\DeclareDocumentCommand\variation{ o g d() }{ 
    \IfNoValueTF{#2}{
        \IfNoValueTF{#3}
        {\updelta \IfNoValueTF{#1}{}{^{#1}}}
        {\mathinner{\updelta \IfNoValueTF{#1}{}{^{#1}}\argopen(#3\argclose)}}
    }
    {\mathinner{\updelta \IfNoValueTF{#1}{}{^{#1}}#2} \IfNoValueTF{#3}{}{(#3)}}
}
\DeclareDocumentCommand\var{}{\variation} 

\def\sdiffd{\mathbf{d}}
\DeclareDocumentCommand\sdifferential{ o g d() }{ 
    \IfNoValueTF{#2}{
        \IfNoValueTF{#3}
        {\sdiffd\IfNoValueTF{#1}{}{^{#1}}}
        {\mathinner{\sdiffd\IfNoValueTF{#1}{}{^{#1}}\argopen(#3\argclose)}}
    }
    {\mathinner{\sdiffd\IfNoValueTF{#1}{}{^{#1}}#2} \IfNoValueTF{#3}{}{(#3)}}
}
\DeclareDocumentCommand\sd{}{\sdifferential}

\DeclareDocumentCommand\detp{}{\opbraces{\determinant{}'}}

\DeclareMathOperator{\volume}{vol}
\DeclareDocumentCommand\vol{}{\opbraces{\volume}}

\DeclareMathOperator{\image}{im}
\DeclareDocumentCommand\im{}{\opbraces{\image}}

\DeclareMathOperator{\cokernel}{coker}
\DeclareDocumentCommand\coker{}{\opbraces{\cokernel}}

\RequirePackage{interval}
\intervalconfig{soft open fences}

\usepackage[noabbrev]{cleveref}

\crefname{subsection}{subsection}{subsections}
\crefname{equation}{}{}

\numberwithin{equation}{section}

\usepackage[style=numeric, eprint=true, natbib, backend=biber, maxbibnames=10, giveninits=true, sorting=none, useprefix=true]{biblatex}
\DeclareFieldFormat*{title}{\emph{#1}}
\DeclareFieldFormat{pages}{#1}
\DeclareFieldFormat{labelalpha}{\textsc{#1}}

\renewbibmacro{in:}{}
\AtEveryBibitem{\clearfield{month}}
\AtEveryCitekey{\clearfield{month}}
\begin{document}

\vspace*{-6em}
\begin{center}
    {\Large\bfseries\scshape{Entanglement in BF theory I: \\  Essential topological entanglement}}
\end{center}

\begin{center}
    \textbf{
        Jackson R. Fliss\textsuperscript{1,\jf} and
        Stathis Vitouladitis\textsuperscript{2,\st}
    }
\end{center}

\begin{center}
    \textbf{1}  Department of Applied Mathematics and Theoretical Physics, \\ University of
    Cambridge, Cambridge CB3 0WA, United Kingdom
    \\[0.5em]
    \textbf{2}  Institute for Theoretical Physics, University of Amsterdam, \\ 1090 GL
    Amsterdam, The Netherlands
    \\[\baselineskip]


    \jf\ \href{mailto:jf768@cam.ac.uk}{\small \sf jf768@cam.ac.uk} \qquad \st\ \href{mailto:e.vitouladitis@uva.nl}{\small \sf e.vitouladitis@uva.nl}
\end{center}

\section*{Abstract}
\vspace{-1em}
We study the entanglement structure of Abelian topological order described by \(p\)-form BF theory in arbitrary dimensions. We do so directly in the low-energy topological quantum field theory by considering the algebra of topological surface operators. We define two appropriate notions of subregion operator algebras which are related by a form of electric-magnetic duality. To each subregion algebra we assign an entanglement entropy which we coin \emph{essential topological entanglement}. This is a refinement to the traditional topological entanglement entropy. It is intrinsic to the theory, inherently finite, positive, and sensitive to more intricate topological features of the state and the entangling region. This paper is the first in a series of papers investigating entanglement and topological order in higher dimensions.

\vspace{10pt}
\tableofcontents
\vspace{10pt}

\section{Introduction}\label{sect:intro}

Quantum entanglement is an invaluable framework for modern theoretical physics. This framework has led to profound insights into quantum information theory, quantum field theory, and quantum gravity.\footnote{See \cite{Faulkner:2022mlp} and references therein for an (obviously) non-exhaustive summary.} Yet some of the most profound applications can be found in the theory of quantum phases of matter.  In particular, in (2+1) dimensional gapped systems the presence of topological order cannot be diagnosed by any local order parameter. Entanglement is a non-local phenomenon. It stands to reason that long-range entanglement can provide a clean signature of topological order in (2+1) dimensions: the celebrated \textquote{topological entanglement entropy} (TEE) \cite{kitaev2006topological,Levin:2006zz}.

Low-energy effective field theories are potent tools for exploring TEE in manifestly universal manner. These are topological quantum field theories (TQFTs), the prototypical example being Chern-Simons theory in (2+1) dimensions.  Topological order in higher-dimensions is expected to be richer: already the discovery of (3+1) dimensional systems displaying \textquote{fracton topological order} \cite{bravyi2011topological,bravyi2013quantum,chamon2005quantum} has broadened our understanding of gapped phases.  Yet even the traditional classification of TQFTs can involve a large set of non-Gaussian interactions which induce richer forms of operator statistics \cite{Ye:2015eba,Tiwari:2016zru}.  It remains a broad open question as to what universal entanglement signatures diagnose and distinguish topological order in higher dimensions.  Here we take modest steps towards understanding this question, focusing on Abelian topological orders described by Abelian BF theory. This focus buys us some muscle: we will be able to make broad statements about Abelian topological order in arbitrary dimensions and quantized on (almost) arbitrary manifolds.\footnote{We do restrict to torsion-free manifolds as well.} We will use this muscle to address two conceptually puzzling aspects of the traditional treatements of TEE.

The first conceptual puzzle we want to address is the area law. Traditional computations of TEE involve an area law stemming from short-distance correlations at the UV scale and to which the TEE appears as a subleading, scale-independent, correction.  Heuristically the scale-independence of this subleading correction is a signal of its universality (however there are subtleties applying this argument to lattice and tensor network models \cite{Zou:2016dck,Kim:2023ydi}).  It is initially surprising that a TQFT, which has a finite dimensional Hilbert space when quantized on a compact surface, can support a divergent entanglement entropy.  However, the area law arises from an explicit addition of UV degrees of freedom when calculating TEE.  These either come in the form of an embedding into a microscopic model (e.g. a lattice gauge theory \cite{Buividovich:2008gq,Buividovich:2008kq,Buividovich:2008yv}, a \textquote{coupled wires} model \cite{Cano:2014pya}, or a tensor network model) or in the form of \textquote{edge-modes} living on an entangling surface \cite{Li:2008kda,chandran2011bulk,Swingle:2011hu,Qi:2012aa}.  These UV degrees of freedom play an important role in calculating entanglement entropy: TQFTs are quantum gauge theories which have a well-known obstruction to factorizing the Hilbert space into local subregions \cite{Buividovich:2008gq,Donnelly:2011hn,Casini:2013rba,Donnelly:2014fua,Soni:2015yga}. In this context the UV degrees of freedom provide an arena, the \textquote{extended Hilbert space,} in which the Hilbert space can be factorized and the entanglement entropy defined. Here we ask if there is another manner for defining entanglement entropy that (i) bypasses invoking UV degrees of freedom, (ii) is strictly topological, and (iii) is commensurate with a finite dimensional Hilbert space in the IR.

There is indeed an alternative for dealing with this obstruction.  In a seminal paper, Casini, Huerta, and Rosabal \cite{Casini:2013rba} illustrated how operator algebras provide a natural definition of entanglement in gauge theories. The lack of Hilbert space factorization manifests itself as a non-trivial center in the algebra of operators associated to a region. Algebraic definitions of entanglement in gauge theories and their relation to the extended Hilbert space have been largely explored in the context of lattice gauge theories \cite{Lin:2018bud,Radicevic:2015sza}.  However, the algebraic approach to entanglement is, in principle, valid even in the continuum. For TQFTs it provides an intrinsically IR avenue for defining entanglement entropy.  I.e. a definition that utilizes only the ground states and operators available at low-energies, without involving UV degrees of freedom, and is strictly finite.  Despite the hotbed of research in entanglement entropy in topological phases and quantum gauge theories, this aspect of topological entanglement has been left relatively unexplored.

The second conceptual puzzle we want to address is \textquote{semi-locality} of the traditional TEE which in (2+1) dimensions involves topological aspects (Betti numbers) intrinsic to the entangling surface; this relation is argued to hold in higher dimensions \cite{Grover:2011fa}. The ground states of topological field theories display extreme long-range entanglement.\footnote{Illustrated, for instance, in \textquote{multi-boundary} set-ups in Chern-Simons theory \cite{Balasubramanian:2016sro}.} It is perhaps then surprising that TEE does not \textquote{sense} farther than the entangling surface itself and is insensitive, say, to how the entangling surface is topologically embedded into the Cauchy slice defining the Hilbert space.

We will address these two conceptual puzzles in this paper. Namely, we consider the operator algebra, $\fA[\Sigma]$, acting on a Cauchy slice, $\Sigma$, that is directly available in Abelian BF theory. The operators generating this algebra are higher-form Wilson surface operators.  Due to the topological nature of the field theory, these surface operators are invariant under deformations and so are naturally associated with homology cycles of $\Sigma$. Clearly this algebra, and any subalgebra, is inherently topological and defined directly in the IR.  However, such operators are not wont to be localized to spatial subregions; as a result, there are potentially large ambiguities in ascribing a subalgebra, $\fA[\aent]$, to region, $\aent$.  We describe two natural choices that can roughly be stated as \textquote{the set of operators that \emph{can} act entirely in $\aent$} and \textquote{the set of operators that \emph{must} act, at least partially, in $\aent$.}  We name these two algebras the \emph{topological magnetic algebra} and the \emph{topological electric algebra}, respectively, for reasons that will become clear in due time. They are related by a form of {\it subregion electric-magnetic duality} which we will make precise below.

We utilize these two notions of subregion algebra to assign an entanglement entropy to ground states in the theory.  This entanglement entropy is by nature (i) topological, and (ii) finite and commensurate with a finite dimensional Hilbert space.  To distinguish it from the traditional TEE appearing as the subleading correction to an area law, we coin\footnote{Competing nomenclatures: intrinsic, core, and boneless topological entanglement. 
} this entropy \emph{essential topological entanglement}, $\cE$. It comes in two forms, $\cE_\t{mag}$ and $\cE_\t{elec}$, and are related by the subregion electric-magnetic duality mentioned above.

Owing to the power of topological field theory, we will be able to evaluate $\cE$ in arbitrary dimensions, on arbitrary surfaces, and  associated to arbitrary regions.  This allows to us to show that $\cE$ is indeed sensitive to more intricate and long-range forms of topology than that of $\pa\aent$ alone: in both forms it depends on topological aspects of $\pa\aent$, $\Sigma$, and how $\pa\aent$ is embedded into $\Sigma$.  This is, again, innate to the operator algebra definition. Operators in $\fA[\aent]$ must, foremost, be operators in $\fA[\Sigma]$.  It is clear then that cycles of $\pa\aent$ that embed to trivial cycles of $\Sigma$ cannot contribute to $\cE$. 

As we will see below, this topological data appears as the coefficient of what can be regarded as the total quantum dimension of surface operators forming $\fA[\Sigma]$ and a coarse measure of non-trivial surface operator braiding. In this regard $\cE$ mimics the traditional TEE, however differs in some important aspects for its use in diagnosing topological order. For instance, as will explain below, the ETE vanishes for the canonical examples defining the TEE on the plane or sphere, ultimately following from the fact that the algebras associated to these spaces must be trivial. It may then appear that $\cE$ is of limited use as a diagnostic of topological order. However, as previously emphasized, the sensitivity of $\cE$ to long-range topological features (beyond the entangling surface) can provide a view on richer features that a topological phase may be sensitive to. We comment on these points in the discussion.

We pause to mention that similar notions to our definition of $\cE$ have appeared in the context of lattice gauge theories by examining the algebras of \textquote{ribbon operators} which are also naturally topological operator algebras \cite{Delcamp:2016eya}.  We are also aware of upcoming work utilizing similar ideas to discuss the \textquote{area operator} in tensor network models of holographic entanglement \cite{AkersSoni}.  However, the focus on these quoted works is on (2+1) dimensional non-Abelian models on spaces and subregions with simple topology.  Our focus on BF theory allows us to work directly in the continuum and deftly incorporate spaces and entangling regions of arbitrary topology, albeit at the expense of working in an Abelian model.  Because of this it is hard to make a direct comparison between these works and ours at this time. We will comment on this further in \Cref{sect:disc}.

Lastly, we also mention that this paper is the first in a series of papers exploring Abelian topological entanglement in higher dimensions.  In a follow up paper \cite{Fliss:2023uiv} we will investigate the \textquote{traditional} TEE in Abelian BF theory by the methods of the extended Hilbert space and the replica path integral.  There we will show that the edge mode spectrum organizes into characters of an infinite dimensional current algebra which generates the entanglement spectrum.  This leads to a TEE of the more traditional variety: area and sub-area laws plus subleading corrections dependent on the topology of $\pa\aent$ alone.

\subsection{Notation}
We will delineate some basic notation for what follows here.

We will consider theories on torsion-free manifolds of spacetime dimension $d$.  We will denote such manifolds collectively as $X$.  Theories will be quantized on $(d-1)$-dimensional manifolds that we will notate as $\Sigma$. We will often call $\Sigma$ the \textquote{Cauchy slice} simply as a term of familiarity and without reference to any causal structure of the TQFT. In what follows we will make use of the notion of a subregion, $\aent$, which will be the closure of a $(d-1)$-dimensional embedded open submanifold of $\Sigma$. We will denote the interior of $\aent$ as $\overline{\aent}\coloneqq\aent\setminus\pa\aent$. We denote the complement $\coaent$ as the closure of $\Sigma\setminus\aent$.  Note that $\aent\cap\coaent=\pa\aent$.

The space of forms of degree $p$, will be denoted $\Omega^p(\cdot)$. Unless stated otherwise these forms are real valued. Cohomology groups will be denoted with their degree placed upstairs, $\H^p(\cdot)$, while homology groups will be denoted with their degree placed downstairs, $\H_p(\cdot)$. Unless otherwise stated, these groups are always defined with integer coefficients. For compact, boundary-less manifolds, we notate the dimensions of the groups by the Betti number, i.e.:
\begin{equation}
	\b_p(\Sigma)\coloneqq\dim\H^p(\Sigma)=\dim\H_p(\Sigma).
\end{equation}
For (co)homology groups on manifolds with boundary or for relative homology groups we will always write the dimension explicitly.

Given a Hilbert space, $\mc H_\Sigma$, defined on a Cauchy slice, $\Sigma$, we will denote the algebra of bounded operators acting on $\mc H_\Sigma$ as $\fA[\Sigma]$. For a subalgebra $\fA_\t{sub}\subset\fA[\Sigma]$ we will denote the commutant as $\co{\left(\fA_\t{sub}\right)}\coloneqq\set{\mc O\in\fA[\Sigma]\suchthat \comm{\mc O}{\mc O'}=0,\ \fall\mc O'\in\fA_\t{sub}}.$

\section{Quantization of BF theory}\label{sect:canonquant}

We begin by introducing the $p$-form Abelian BF theory, on a \(d\)-dimensional, torsion-free manifold \(X\), with action
\begin{equation}\label{eq:BFact}
	S_\BF\qty[A,B] \coloneqq \frac{\bbK^{\sfI\sfJ}}{2\pi} \int_X B_\sfI\w\dd{A_\sfJ}.
\end{equation}
In the above \(A_\sfI\in \Omega^{p}(X)\) and \(B_\sfI\in \Omega^{d-p-1}(X)\) are vectors of \(p\)-  and \((d-p-1)\)-form gauge fields respectively. We will take $p\neq 0,d-1$.\footnote{We expect much of what follows to morally hold true in these special cases, however some technical details of our proofs would need to be altered.} We have also allowed a possible square, integer, non-degenerate --- but not necessarily symmetric --- \(\bbK\)-matrix of rank \(\kappa\). For notational simplicity we will drop the indices, unless it is necessary.  In \Cref{app:precBF} we provide a more careful treatment of BF theory, allowing for manifolds with torsion. The action \Cref{eq:BFact} possesses a gauge redundancy of the form
\begin{equation}\label{eq:redundancy}
	\var{A}=\dd\alpha\qq{and} \var{B}=\dd{\beta},
\end{equation}
where $\alpha\in\Omega^{p}(X)$ and $\beta\in\Omega^{d-p-2}(X)$.

Let us suppose that $X$ possesses a boundary and discuss the quantization of the theory on $\pa X$.  Much of this procedure follows that of \cite{Bergeron:1994ym} and \cite{Tiwari:2016zru}, however we provide these details for completeness.  We begin with the classical symplectic structure.  The variation of the action takes the form
\begin{equation}
	\var S_\BF\qty[A,B]=\int_X\left(\var B\w \eom[A]+\var A\w\eom[B]\right)+\int_{\pa X}{\boldsymbol\vartheta}\qty[A,B;\var A,\var B],
\end{equation}
where the classical equations of motion are flatness conditions:
\begin{equation}
	\eom[A]=\frac{\bbK}{2\pi}\dd{A}=0\qquad\qquad \eom[B]=(-1)^{(d-p)(p+1)}\frac{\trans\bbK}{2\pi}\dd{B}=0.
\end{equation}
The boundary term defines the symplectic potential, ${\boldsymbol\vartheta}$:
\begin{equation}\label{eq:symppot}
	\int_{\pa X}{\boldsymbol\vartheta}\qty[A,B;\var A, \var B]\coloneqq(-1)^{d-p-1}\frac{\bbK}{2\pi}\int_{\pa X}B\wedge\var A,
\end{equation}
where the pullback along the embedding map, \(\ib:\pd X\hookrightarrow X\) is implicitly understood above.  We see this theory is already in canonical, or Darboux, form, ${\boldsymbol\vartheta}={\cmom}\wedge\star_{\scriptscriptstyle \pd X}\var{\cpos}$, with
\begin{equation}
	\cpos = A \qq{and} {\cmom}=(-1)^{d-p-1}\frac{\bbK}{2\pi}\star_{\scriptscriptstyle \pd X}B,
\end{equation}
which is consistent with fixing $A$ as a boundary condition.  We can switch the role of $(\star_{\scriptscriptstyle \pd X}B,A)\sim({\cmom},{\cpos})$ to $(\star_{\scriptscriptstyle \pd X}B,A)\sim({\cpos},{\cmom})$ by the inclusion of the boundary action
\begin{equation}
	S_{\pa}^{\text{alt.}}[A,B]=(-1)^{d-p}\frac{\bbK}{2\pi}\int_{\pa X}B\wedge A,
\end{equation}
but we will work in the former quantization scheme.  The symplectic form on $\pa X$, given by the variation of $\int_{\pa X}{\boldsymbol\vartheta}$, is
\begin{equation}\label{eq:sympBA}
	{\boldsymbol\Omega}_{\pa X}=(-1)^{d-p-1}\frac{\bbK}{2\pi}\int_{\pa X}\var B\wedge \var A.
\end{equation}
This symplectic form is degenerate due to gauge variations.  We will take care of this soon below.

We will quantize the BF theory on a $(d-1)$-dimensional manifold, \(\Sigma\), by performing the path-integral on \(X = \R\times \Sigma\). Here, \(\R\) is coordinatized by \(t\) and the Cauchy slice at time \(t\) is represented by \(\set{t}\times \Sigma\). The path-integral measure is formally given by
\begin{equation}
	\dd{\mu}(A,B) = \frac{\DD{A}\ \DD{B}}{\vol(\cG_p)\vol(\cG_{d-p-1})}\;\ex{\ii S_\BF[A,B]},
\end{equation}
where \(\cG_p\) and \(\cG_{d-p-1}\) are the gauge groups for the redundancies \Cref{eq:redundancy}. On top of it, it includes a sum over non-trivial bundles. For a full definition of the measure and the gauge groups, we refer the reader to \Cref{app:precBF}. Currently, we consider the case where \(\pd \Sigma = \emptyset\) (we will revisit the case with boundaries in a follow-up paper). Additionally, let \(\iota_\Sigma : \Sigma\hookrightarrow X\) be the embedding of \(\Sigma\) into \(X\).

We can express \(A\) and \(B\) as \(A = A_0 + a\) and \(B = B_0 + b\) respectively, where
\begin{align}
	\iota_\Sigma^* A_0 = 0 & \qquad\eqv\qquad \iota_\Sigma^* A  = \iota_\Sigma^* a \nn
	\iota_\Sigma^* B_0 = 0 & \qquad\eqv\qquad \iota_\Sigma^* B  = \iota_\Sigma^* b. \nonumber
\end{align}
To make the decomposition clearer we can use coordinates \(\set{x_m}_{m=1}^{d-1}\) for \(\Sigma\) which gives us:
\begin{align}
	A & = \qty(A_0)_{m_1\cdots m_{p-1}} \dd{t}\w\dd{x}^{m_1} \w\cdots\w \dd{x}^{m_{p-1}} + a_{m_1\cdots m_p} \dd{x}^{m_1} \w\cdots\w \dd{x}^{m_p} \qq{and}         \\
	B & = \qty(B_0)_{m_1\cdots m_{d-p-2}} \dd{t}\w\dd{x}^{m_1} \w\cdots\w \dd{x}^{m_{d-p-2}} + b_{m_1\cdots m_{d-p-1}} \dd{x}^{m_1} \w\cdots\w \dd{x}^{m_{d-p-1}}.
\end{align}
In these coordinates, let us also write \(\dd = \dd{t}\w\pd_t + \dd{x^m}\w\pd_m \eqqcolon \dd_\R + \sd\). Integrating \Cref{eq:BFact} by parts and utilizing the fact that \(\dd_\R A_0 = 0\) and \(\dd_\R B_0 = 0\) (since they involve \(\dd{t}\w\dd{t}\w\cdots=0\)), we arrive at
\begin{equation}
	S_\BF\qty[A_0+a,B_0+b] = \frac{\bbK}{2\pi}\int_X \qty\Big((-1)^{d-p}\sd{b}\w A_0 + B_0 \w \sd{a} + b\w\dd a). \label{eq:act-split}
\end{equation}
It is easy to see that \(A_0\) and \(B_0\) act as Lagrange multipliers enforcing the \(\Sigma\)-flatness of \(a\) and \(b\):
\begin{equation}\label{eq:GaussLaw}
	\sd a=\sd b=0.
\end{equation}
We will refer to \Cref{eq:GaussLaw} as the \textquote{Gauss law} constraints.  Using the property \Cref{eq:property} of the path-integral measure we can write
\begin{equation}
	\dd{\mu(A,B)} = \dd{\mu(A_0,B_0)}\; \dd{\mu(a,b)}\ \ex{\ii S_\BF[A_0,b]}\ \ex{\ii S_\BF[a,B_0]}
\end{equation}
and performing the integrals over \(A_0\) and \(B_0\) we get \(\dd{\mu(a,b)}\delta\qty[\sd{a}]\delta\qty[\sd{b}]\). The delta-functions force \(a\) and \(b\) to be closed under \(\sd\); Hodge decomposition implies, then, that
\begin{equation}\label{eq:abhodge}
	a = \sd{\psi} + \theta, \qq{and} b=\sd{\chi} + \phi,
\end{equation}
for some \(\psi\in \Omega^{p-1}(X)\), \(\theta\in\iota_\Sigma^* \H^p(\Sigma)\), and \(\chi\in \Omega^{d-p-2}(X)\), \(\phi\in\iota_\Sigma^* \H^p(\Sigma)\). This results into the path-integral measure:
\begin{align}
	\dd{\mu(a,b)}\ \delta\qty[\sd{a}]\delta\qty[\sd{b}]
	 & = \frac{\DD{\psi}\DD{\chi}\DD{\phi}\DD{\theta}}{\vol(\cG_p)\vol(\cG_{d-p-1})}\ \exp(\frac{\ii \bbK}{2\pi}\int_X \phi\w\dd_\R {\theta}).
\end{align}
The integral over \(\psi\) and \(\chi\) over the volumes of the gauge groups yields the Ray--Singer torsion of the manifold \cite{Blau:1989bq,Gegenberg:1993gd},\footnote{for a modern exposition see also \cite{blauMassiveRaySingerTorsion2022}} and so we simply get
\begin{equation}\label{eq:BFPI2}
	\parti_\BF[X] = \qty(\int\DD{\phi}\DD{\theta}\ \exp(\frac{\ii \bbK}{2\pi}\int_X \phi\w\dd_\R {\theta}))\ \mathrm{T}_\t{RS}[X]^{(-1)^{p-1}},
\end{equation}
where \(\mathrm{T}_\t{RS}[X]\) is the Ray--Singer torsion:
\begin{equation}
	\mathrm{T}_\t{RS}[X] \coloneqq \prod_{k=0}^{d} \qty(\frac{\qty(\detp\lapl_k)^{k}}{\det\bbG_k})^{\Half(-1)^{k+1}},
\end{equation}
with \(\lapl_k\) being the Laplacian on the space of \(k\)-forms on \(X\), and \(\bbG_k\) the metric in the space of harmonic \(k\)-forms, defined as follows. Let $\set{\tau^{(k)}_\sfi}_{\sfi=1}^{\b_k(\Sigma)}$ be the topological basis of harmonic $k$-forms, with \({\b_k(\Sigma)}\) the $k^\t{th}$ Betti number of $\Sigma$. This basis is defined such that given a basis of \(k\)-cycles \(\set{\eta^\sfi_{(k)}\in\H_k(\Sigma)}_{\sfi=1}^{\b_k(\Sigma)}\), there is a unique harmonic representative, $\tau_\sfi^{(k)}$, of each cohomology class in \(\H^k(\Sigma)\), such that
\begin{equation}\label{eq:cyclecocycleortho}
	\int_{\eta^\sfi_{(k)}} \tau_\sfj^{(k)} = \delta^\sfi_{\sfj}.
\end{equation}
It is in terms of this basis that the matrices \(\bbG_k\) above are defined. Explicitly:
\begin{equation}
	\qty[\bbG_k]_{\sfi\sfj} \coloneqq \int_\Sigma \tau_\sfi^{(k)} \w\star \tau_\sfj^{(k)},
\end{equation}
where \(\star\) is the Hodge-star on \(\Sigma\). Before moving on let us make a quick digression to mention that the inverse of \(\bbG_k\) is the linking matrix
\begin{equation}
	[\bbG_k]_{\sfi\sfj} [\bbL_k]^{\sfj\sfk} = \delta^\sfk_\sfi,
\end{equation}
which can be alternatively defined as an oriented intersection number in the following way. Let us pick a basis of $k$-cycles $\set{\eta^\sfi}_{\sfi=1}^{{\b_k(\Sigma)}}$ of $\H_k(\Sigma)$ and a basis of $(d-k-1)$-cycles $\set{\sigma^\sfi}_{\sfi=1}^{{\b_k(\Sigma)}}$ of $\H_{d-k-1}(\Sigma)$.  The transversal intersection of $\eta^\sfj$ and $\sigma^\sfi$ in $\Sigma$ is a zero-dimensional manifold (that is, a collection of points) and $[\bbL_k]^{\sfi\sfj}$ counts the number of points signed by their orientation:
\begin{equation}
	[\bbL_k]^{\sfi\sfj}\equiv\bbL_k\qty(\eta^\sfi,\sigma^\sfj)\coloneqq\int_{\eta^\sfi\cap\sigma^\sfj}1.
\end{equation}
Let us focus on the remaining path-integral in \Cref{eq:BFPI2}, which is the quantum mechanics for the large-gauge degrees of freedom, \(\phi\) and \(\theta\). We can expand \(\phi\) and \(\theta\) in terms of the basis $\set{\tau_\sfi}_{\sfi=1}^{{\b_p(\Sigma)}}$. Namely
\begin{equation}
	\phi(t,x) = \phi^\sfi(t)\, \tau_\sfi \qq{and} \theta(t,x) = \theta^\sfj(t) \star \tau_\sfj.
\end{equation}
Note that \Cref{eq:cyclecocycleortho} with \Cref{eq:abhodge} implies
\begin{equation}
	\theta^\sfi=\int_{\eta^\sfi}a\qquad\qquad\phi^\sfi=\int_{\sigma^\sfi}b,
\end{equation}
in terms of our original field variables.
Since \(\phi^\sfi(t)\) and \(\theta^\sfj(t)\) are circle-valued functions on \(\R\) they are identified with
\begin{equation}
	\phi^\sfi(t) \sim \phi^\sfi(t) + 2\pi \qq{and} \theta^\sfj(t)\sim \theta^\sfj(t) + 2\pi.
\end{equation}
All in all the action reduces to
\begin{equation}
	S_\BF^\t{eff}[\phi,\theta] = \frac{\bbK}{2\pi} \qty[\bbG_p]_{\sfi\sfj} \int_{\R} \phi^\sfi \w \dd_\R \theta^\sfj.
\end{equation}
This is a simple quantum mechanical system whose symplectic form reads (reinstating the \(\sfI,\sfJ\) indices)
\begin{equation}\label{eq:sympphitheta}
	{\boldsymbol\Omega}_\Sigma = \frac{(-1)^{d-p-1}}{2\pi} \bbK^{\sfI\sfJ} [\bbG_p]_{\sfi\sfj} \var{\phi^\sfi_\sfI} \w\var{\theta^\sfj_\sfJ}.
\end{equation}
which is the restriction of \Cref{eq:sympBA} to $\phi^\sfi_\sfI$ and $\theta^\sfj_\sfJ$.

Given the our interpretation of $A=\cpos$ coming from the symplectic potential, \Cref{eq:symppot}, we will identify \(\theta^\sfj_\sfJ\) as \textquote{positions} and \(\phi^\sfi_\sfI\) as \textquote{momenta}.\footnote{More correctly handling the index placement, the momenta are \(\cmom_\sfj^\sfJ=(-1)^{d-p-1}\frac{\bbK^{\sfI\sfJ}}{2\pi}[\bbG_p]_{\sfi\sfj}\phi^\sfi_\sfI\).}  Passing from Poisson brackets to commutators, promoting \(\phi\) and \(\theta\) to operators, we arrive at
\begin{equation}
	\comm{\hat{\phi}^\sfi_\sfI}{\hat{\theta}^\sfj_\sfJ} = 2\pi\ii\; (-1)^{d-p-1}\qty[\bbK^{\perp}]_{\sfI\sfJ} [\bbL_p]^{\sfi\sfj},    \label{eq:ccr}
\end{equation}
where \(\bbK^{\perp}\coloneqq \inv{\qty(\trans{\bbK})}\) is the inverse transpose. Hereafter we will drop the index $p$ and the square brackets from $\bbL_p$, and $\bbG_p$ for conciseness.

Since these operators are \(\U(1)\)-valued, we should exponentiate them to construct gauge-invariant Wilson surface operators:
\begin{align}
	\hW_{\eta^\sfj}^{\vec w_\sfj}   & \coloneqq \exp(w^\sfJ_\sfj\ \hat{\theta}^\sfj_\sfJ) = \exp(\int_{\eta^\sfj} w_\sfj^\sfJ\ a_\sfJ)   \\
	\hV_{\sigma^\sfi}^{\vec v_\sfi} & \coloneqq \exp(v^\sfI_\sfi \ \hat{\phi}^\sfi_\sfI) = \exp(\int_{\sigma^\sfi} v^\sfI_\sfi\ b_\sfI),
\end{align}
where $\set{w_\sfj^\sfJ}_{\sfj\in\set{1,\ldots,{\b_p(\Sigma)}}}^{\sfJ\in\set{1,\ldots,\kappa}}$ and $\set{v^\sfI_\sfi\vphantom{(q_\t A)_\sfj^\sfJ}}_{\sfi\in\set{1,\ldots,{\b_p(\Sigma)}}}^{\sfI\in\set{1,\ldots,\kappa}}$ are ${\b_p(\Sigma)}\times\kappa$ collections of integers.  These surface operators are defined with respect to fixed bases of homology $p$- and $(d-p-1)$-cycles, $\set{\eta^\sfj}_{\sfj=1}^{\b_p(\Sigma)}$ and $\set{\sigma^\sfi}_{\sfi=1}^{\b_p(\Sigma)}$, respectively; however it is easy to verify that they are homotopy invariants when acting on gauge-invariant states due to the Gauss-law constraints, \Cref{eq:GaussLaw}, and thus well defined on homology classes.

\subsection{The algebra of Wilson surface operators}

The Wilson surface operators constructed above satisfy a \textquote{clock algebra} which can be easily found using the canonical commutation relations \Cref{eq:ccr}:
\begin{equation}\label{eq:clockalgcomponents}
	\hV_{\sigma^\sfi}^{\vec v_i}\ \hW_{\eta^\sfj}^{\vec w_j} = \ex{2\pi\ii\; (-1)^{d-p-1}\;w_\sfj^\sfJ v_\sfi^\sfI(\bbK^\perp)_{\sfI\sfJ}\bbL^{\sfi\sfj}}\ \hat{\W}_{\eta^\sfj}^{\vec w_\sfj}\ \hV_{\sigma^\sfi}^{\vec v_\sfj}.
\end{equation}
The $\hW$'s commute amongst each other as do the $\hV$'s. From the above algebra we can clearly see that \(\vec w_\sfj\) and \(\vec v_\sfi\) give the same algebra as \(\vec w_\sfj + m\cdot \bbK\) and \(\vec v_\sfi + \bbK\cdot m'\), for arbitrary \(m,m'\in\Z^\kappa\). The entire algebra is, then, generated by operators labelled by charges in the lattices
\begin{equation}
	\vec w_\sfj\in \lattice_{A} \coloneqq \Z^\kappa\big/\im\trans{\bbK} \qq{and} \vec v_\sfi\in \lattice_{B} \coloneqq \Z^\kappa\big/\im\bbK.
\end{equation}
It will be notationally useful to collect $w^\sfJ_\sfj$ and $v^\sfI_\sfi$ as the components of a ${\b_p(\Sigma)}\times\kappa$-dimensional integer vectors denoted as $\fw$ and $\fv$, respectively.  Also for notational convenience, we will define an inner product on these vector spaces as
\begin{equation}
	\Gamma\qty(\fv,\fw)\coloneqq2\pi(-1)^{d-p-1}\fv\cdot\left(\bbK^\perp\otimes\bbL\right)\cdot \fw=2\pi(-1)^{d-p-1}w^\sfJ_\sfj\,v^\sfI_\sfi(\bbK^\perp)_{\sfI\sfJ}\bbL^{\sfi\sfj}.
\end{equation}
Then \Cref{eq:clockalgcomponents} can be written succinctly as
\begin{equation}\label{eq:WValg}
	\hV^{\fv}\hW^{\fw}=\ex{\ii\Gamma\qty(\fv,\fw)}\hW^{\fw}\hV^{\fv},
\end{equation}
with
\begin{equation}
	\hV^{\fv}=\prod_{i=1}^{\b_p(\Sigma)}\hV^{\vec v_i}_{\sigma^i},\qquad \hW^{\fw}=\prod_{i=1}^{\b_p(\Sigma)}\hW^{\vec w_j}_{\eta^j}.
\end{equation}
In this notation, the $\hW$'s and $\hV$'s satisfy an Abelian fusion algebra
\begin{equation}\label{eq:WVfusionalg}
	\hW^{\fw}\hW^{\fw'}=\hW^{\fw+\fw'} \qq{and} \hV^{\fv}\hV^{\fv'}=\hV^{\fv+\fv'},
\end{equation}
where it is understood that the sums are taken in the respective lattices, $(\lattice_{A})^{\b_p(\Sigma)}$ and $(\lattice_{B})^{\b_p(\Sigma)}$.

\subsubsection*{Constructing states}

To construct the states on \(\cH_\Sigma\) we pick a maximal set of commuting operators and use the space of their eigenvectors. For that we can use either \(\set{\hat{\phi}^\sfj}_{\sfj=1}^{{\b_p(\Sigma)}}\) or \(\set{\hat{\theta}^\sfi}_{\sfi=1}^{{\b_p(\Sigma)}}\). We will first use the basis given by \(\set{\hat{\theta}^\sfi}\) eigenvectors which is morally consistent with fixing $a$ as a boundary condition. To construct the states systematically, we will first define a fiducial state \(\ket{0}\) annihilated by all \(\set{\hat\theta^\sfi}\). This is an eigenstate of \(\hat{\W}^{\vec w_\sfj}_{\eta^\sfj}\) with all eigenvalues one:
\begin{equation}
	\hW^{\fw}\ket{0}=\ket{0}.
\end{equation}
We will call this state the \textquote{$p$-surface operator condensate,} or \textquote{the condensate} when the context is clear. We then use \(\hV^{\vec v_\sfi}_{\sigma^\sfi}\) as raising operators.  A general ground state will then be given by
\begin{equation}
	\ket{\fv}\coloneqq \hV^{\fv}\ket{0}=\prod_{\sfi=1}^{\b_p(\Sigma)}\hV_{\sigma^\sfi}^{\vec v_\sfi}\ket{0}.
\end{equation}
for any integer vector, $\fv\in(\lattice_B)^{\b_p(\Sigma)}$.  These are indeed eigenstates of $\hW^{\fw}$ with eigenvalue
\begin{equation}
	\hW^{\fw}\ket{\fv}=\ex{\ii\Gamma\qty(\fv,\fw)}\ket{\fv}.
\end{equation}
Additionally, since
\begin{equation}
	\braket{0}{\fv}=\mel{0}{\hW^{\fw}}{\fv}=\ex{\ii\Gamma\qty(\fv,\fw)}\braket{0}{\fv},
\end{equation}
for any $\fw\in(\lattice_A)^{\b_p(\Sigma)}$ and since $\bbK$ and $\bbL$ (and thus $\Gamma$) are non-degenerate,
\begin{align}\label{eq:latticedelta}
	\braket{0}{\fv}={\boldsymbol\delta}_{\fv}\coloneqq \prod_{\sfi=1}^{{\b_p(\Sigma)}}\delta_{\vec v_\sfi}^{(\lattice_B)}, \qq{with} \delta_{\vec v}^{(\lattice_B)}=\begin{cases}1 & \vec v=0\mod\im\bbK \\
             0 & \text{otherwise}\end{cases}.
\end{align}

This fact, coupled with $(\hV^{\fv})^\dagger=\hV^{-\fv}$, and the fusion algebra, \Cref{eq:WVfusionalg}, implies the full-orthonormality of $\cH_\Sigma=\spn\set{\ket{\fv}}$.  The dimension of the Hilbert space is
\begin{equation}
	\dim\cH_\Sigma = \dim_\Z \qty(\lattice_{B})^{{\b_p(\Sigma)}} =\dim_\Z(\Lambda_A)^{\b_p(\Sigma)} = \abs{\det\bbK}^{{\b_p(\Sigma)}}. \label{eq:dimH}
\end{equation}
We note that the quantization can be repeated in a wholly similar procedure by using $\hW^{\fw}$ as ladder operators (this builds a Hilbert space of eigenvectors of $\hV^{\fv}$) to arrive a Hilbert space of the same dimension. We will call the isomorphism of the Hilbert spaces built on $p$- and $(d-p-1)$-surface operator condensates, \emph{electric-magnetic duality} in this context.\footnote{This duality is simply a statement that the Hilbert space built on the $p$-surface operator condensate is of equal dimension to the Hilbert space built on the $(d-p-1)$-surface operator condensate. They are automatically isomorphic. This is a simple consequence of Hodge duality on $\Sigma$.} Below we will describe how this duality can be refined to a notion of \emph{subregion electric-magnetic duality}.

\section{Subregion algebras and essential topological entanglement}\label{sect:OAent}

We now move to the main act of this paper: how to associate subregion entanglement entropy to this theory after we have \textquote{integrated out} all of the local degrees of freedom.  We will do so in the algebraic approach.  We briefy remind the reader of the broad features of this approach.

Starting with a region, $\aent\subset\Sigma$, one associates a subalgebra, $\fA[\aent]\subset\fA[\Sigma]$, of the operators which act naturally on $\aent$.  The commutant of $\fA[\aent]$ is then associated to the complement of $\aent$: $\fA[\coaent]=\co{\left(\fA[\aent]\right)}$.  Given a state,\footnote{We will generally call density matrices \textquote{states} regardless of their purity.} $\rho$, one can reduce it to $\fA[\aent]$: i.e. $\rho_\aent$ is the unique Hermitian and trace-normalized element of the subregion algebra, $\fA[\aent]$, reproducing the expectation values of all $\mc O_\aent\in\fA[\aent]$. The von Neumann entropy of this reduced density matrix then provides an algebraic definition of the entanglement entropy of $\rho$ reduced to $\aent$:
\begin{equation}\label{eq:algentdef}
	\alge{\fA[\aent]}[\rho]=\vne[\rho_\aent]\coloneqq-\Tr(\rho_\aent\log\rho_\aent).
\end{equation}
This situation is complicated in theories with gauge invariance. The non-local manner in which gauge constraints are applied to states manifests itself in a non-trivial center in the subregion algebra: $\fZ[\aent]=\fA[\aent]\cap\fA[\coaent]$. Operators generating $\fZ[\aent]$ can be simultaneously diagonalized.  The state, $\rho$, and subsequenty, the reduced state, $\rho_\aent$, can be decomposed with respect to the eigenspaces of the operators generating $\fZ[\aent]$:
\begin{equation}
	\rho_\aent=\bigoplus_\alpha\lambda_{(\alpha)}\rho_\aent^{(\alpha)},
\end{equation}
where $\alpha$ labels the eigenspaces. The $\rho_\aent^{(\alpha)}$ can be individually trace-normalized and so $\sum_\alpha\lambda_{(\alpha)}=1$. This leads to a refinement of \Cref{eq:algentdef} where the algebraic entanglement entropy naturally splits into the weighted sum of von Neumann entropies of reduced density matrix projected to fixed eigenspaces plus the Shannon entropy of the probability distribution given by $\set{\lambda_{(\alpha)}}$:
\begin{equation}\label{eq:algentref}
	\alge{\fA[\aent]}[\rho]=\sum_{\alpha}\lambda_{\left(\alpha\right)}\vne\left[\rho_\aent^{\left(\alpha\right)}\right]-\sum_{\alpha}\lambda_{\left(\alpha\right)}\log\lambda_{\left(\alpha\right)}.
\end{equation}
We will make these broad features explicit in what follows and show that the Shannon contribution takes a universal, topological, form. Before doing so we will first need to define a notion of a subregion algebra, $\fA[\aent]$. Given the topological nature of the operators in $\fA[\Sigma]$, we will take care to define it in a manifestly topological manner below.\footnote{As a benefit to these definitions applied to $\fA[\Sigma]$ and its subalgebras: these are all Type I von Neumann algebras acting on finite dimensional spaces.  As such there is no subtlety in defining traces, reduced density matrices, and von Neumann entropies.}

\subsection{Topological subregion algebras}\label{sect:topsubalgs}

We begin by regarding \(\Sigma\) as the union of two, otherwise disjoint, submanifolds \(\Sigma = \aent\sqcup_{\pa\aent} \co{\aent}\) sharing a common boundary, $\pa\aent$.  We will assign an operator algebra, $\fA[{\aent}]$, to \(\aent\) and $\fA[{\co\aent}]$ is then defined as the commutant of $\fA[\aent]$.  There is some ambiguity in this assignment; in what follows we will assign this in a \textquote{natural} way.  Since our operators in this theory are only defined up to homotopy, however, there may be multiple \textquote{natural} ways to associate an algebra to $\aent$.  Different choices of subregion algebra may result in different centers and different definitions of the entanglement entropy.

Let us introduce the following notations. Suppose that \(M\) is a submanifold of \(\Sigma\) and let \(i^M: M\hookrightarrow \Sigma\) be the embedding map. We can use \(i^M\) to push-forward homology groups: \(i^M_\blob:\H_\blob(M)\to\H_\blob(\Sigma)\).\footnote{For notational convenience we will avoid indexing the push-forward with an asterisk or a hash, as is common in the mathematical literature and we will index it solely by the rank of the homology groups it is connecting.} In what follows, given \(\alpha\in\H_k(\Sigma)\) and a map as above, we will denote \(\alpha\hin M\), iff \(\alpha\in\im i^M_k\). In words, $\alpha\hin M$ says that $\alpha$ is continuously deformable within $\Sigma$ to a cycle completely contained in $M$.  Similarly, we denote \(\alpha\nothin M\), for an \(\alpha\in \H_k(\Sigma)\), iff \(\alpha\in\coker i^M_k\). In words, \(\alpha\nothin M\) is not continuously deformable within \(\Sigma\) to a cycle completely contained in \(M\). We will alternate between the \(\im/\coker\) and \(\hin/\nothin\) notation freely. To state the results of the following sections up-front, there are two natural algebras associated to $R$:
\begin{description}
	\item[1. Topological magnetic algebra]
		\begin{align}
			\fA_\t{mag}[\aent] & \coloneqq \fU\set{\hat{\W}_{\eta^\sfi}^{\vec w_\sfi},\hV_{\sigma^\sfj}^{\vec v_\sfj} \suchthat \eta^\sfi\in\im i^\aent_p,\ \sigma^\sfj\in\im i^\aent_{d-p-1}} \\
			                   & \equiv \fU\set{\hat{\W}_{\eta^\sfi}^{\vec w_\sfi},\hV_{\sigma^\sfj}^{\vec v_\sfj} \suchthat \eta^\sfi,\sigma^\sfj\hin \aent},
		\end{align}
		where \(\fU\set{\cdot}\) denotes the universal enveloping algebra. This algebra consists of all surface operators deformable to being completely contained in $\aent$.
	\item[2. Topological electric algebra]
		\begin{align}
			\fA_\t{elec}[R] & \coloneqq \fU\set{\hat{\W}_{\eta^\sfi}^{\vec w_\sfi},\hV_{\sigma^\sfj}^{\vec v_\sfj} \suchthat \eta^\sfi\in\coker i^{\coaent}_p,\ \sigma^\sfj\in\coker i^{\coaent}_{d-p-1}} \\
			                & \equiv \fU\set{\hat{\W}_{\eta^\sfi}^{\vec w_\sfi},\hV_{\sigma^\sfj}^{\vec v_\sfj} \suchthat \eta^\sfi,\sigma^\sfj\nothin \coaent}.
		\end{align}
		This algebra consists of all surface operators that are not deformable to being completely contained in $\co\aent$.
\end{description}
As we will soon explain both of these algebras have non-trivial centers, $\fZ[\aent]$, which we name the \emph{topological magnectic center} and \emph{topological electric center}, respectively. We can seek centerless operator algebras by either systematically removing operators in $\fZ[\aent]$ from $\fA[\aent]$ or by systematically adding operators to $\fA[\aent]$ that do not commute with operators in $\fZ[\aent]$. We do so in \cref{app:centerless}. There we show that it results in two centerless algebras, that have more tenuous relationships to their underlying subregion, while additionally ground states have trivial entanglement with respect to these algebras.

These different choices of subregion algebras are illustrated in  \Cref{fig:3t}.

\begin{figure}[!th]
	\centering
	\def\svgwidth{0.6\textwidth}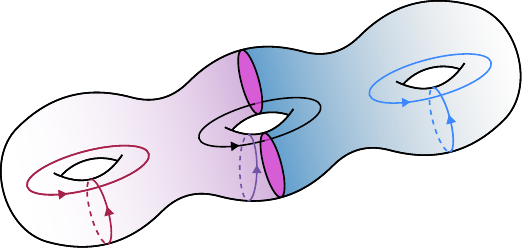
	\def\svgwidth{0.6\textwidth}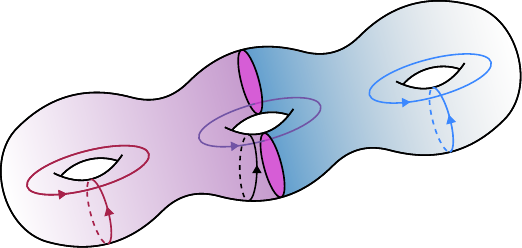
	\caption{The algebra on a 3-torus, generated by a basis of operators along longitude and meridian cycles, $\set{\ell_i,m_i}_{i=1,2,3}$.  To the region, $\aent$, depicted in pink, we associate an algebra $\fA[\aent]$ generated by cycles depicted in \textcolor{rred}{red} and \textcolor{ppurple}{purple}.  The commutant, $\fA[\coaent]$, is generated by cycles depicted in \textcolor{bblue}{blue} and \textcolor{ppurple}{purple}.  The center, $\fZ[\aent]$, is generated by cycles depicted in \textcolor{ppurple}{purple}. Cycles generated neither in $\fA[\aent]$ nor $\fA[\coaent]$ are depicted in \textbf{black}. (Top) The topological magnetic algebra, $\fA_\t{mag}[\aent]$.  (Middle) The topological electric algebra, $\fA_\t{elec}[\aent]$.
    }
	\label{fig:3t}
\end{figure}

\subsubsection{The topological magnetic algebra}

Let us begin the discussion with the topological magnetic algebra
\begin{equation}
	\fA_\t{mag}[\aent] \coloneqq \fU\set{\hat{\W}_{\eta^\sfi}^{\vec w_\sfi},\hV_{\sigma^\sfj}^{\vec v_\sfj} \suchthat \eta^\sfi,\sigma^\sfj\hin \aent},
\end{equation}
The algebra associated with \(\co{\aent}\) is the commutant of \(\fA_\t{mag}[\aent]\). It is clear that the operators that can commute with \(\fA_\t{mag}[\aent]\) are precisely those that cannot link homology cycles in \(\aent\).  This is equivalent to the following.
\begin{equation}\label{eq:claimmag}
	\co{\qty(\fA_\t{mag}[\aent])} = \fU\set{\hat{\W}_{\eta^\sfi}^{\vec w_\sfi},\hV_{\sigma^\sfj}^{\vec v_\sfj} \suchthat  \eta^{\sfi},\sigma^{\sfj}\hin \coaent}\equiv\fA_\t{mag}\qty[\co\aent].
\end{equation}
The proof of this claim is given in \Cref{app:zm}.

As alluded to above, this algebra has a center, $\fZ_\t{mag}$, which is generated by surface operators lying within the entangling surface, $\pa\aent$, itself:
\begin{equation}
	\fZ_\t{mag}[\aent]\coloneqq \fA_\t{mag}[\aent] \cap \fA_\t{mag}[{\co{\aent}}] = \fU\set{\hat{\W}_{\eta^\sfi}^{\vec w_\sfi},\hV_{\sigma^\sfj}^{\vec v_\sfj} \suchthat \eta^\sfi,\sigma^\sfj \hin \pd \aent}.
\end{equation}
The heuristic argument is simple: $\fZ_\t{mag}[\aent]\subset\fA_\t{mag}[\aent]$ since any cycle, $\eta^\sfi\hin\pa\aent$ or $\sigma^\sfj\hin\pa\aent$, can be deformed \textquote{slightly inward} along the flow of an inward-pointing normal vector to be contained completely in $\aent$.  Similarly by flowing in the other direction, \textquote{slightly outward,} any cycle can be contained completely in $\coaent$ and so $\fZ_\t{mag}[\aent]\subset\fA_\t{mag}[\coaent]$. Perhaps one might question if operators in $\fZ_\t{mag}$ actually commute with themselves (as required for $\fZ_\t{mag}$ to be a center).  A potential puzzle arises because $p$-cycles and $(d-p-1)$-cycles have no notion of intersection numbers as defined intrinsically on $\pa\aent$: the intersection of a $p$-cycle and a $(d-p-1)$-cycle on a $d-2$ dimension manifold is not a collection of points, but instead itself a $1$-dimensional manifold.  The key here is the algebra $\fZ_\t{mag}$ is not defined intrinsically on $\pa\aent$ but instead up to homotopy in $\Sigma\supset\pa R$.  It is then clear that all $(d-p-1)$-cycles, $\sigma^\sfj\hin\pa\aent$, can be deformed to have zero linking number with $\eta^\sfi\hin\pa\aent$ by evolving them slightly along an outward pointing normal vector.

We call $\fZ_\t{mag}$ the \emph{topological magnetic center} because of its similarity to the magnetic center of 2+1 dimensional lattice gauge theories \cite{Casini:2013rba,Delcamp:2016eya}, generated by line and/or ribbon operators wrapping $\pa R$.  Here the interplay of the dimensionality, $d$, with the degrees of the gauge fields, $p$ and $d-p-1$, allow for a richer flavour of magnetic center, generated by topological operators of different dimension. Additionally, the topological magnetic center defined here is sensitive to the bulk topology while the magnetic center appearing in \cite{Casini:2013rba} is only sensitive to the intrinsic topology of $\pa\aent$.  Namely, an operator can only appear in $\fZ_\t{mag}$ if its defining cycle is also non-trivial as a cycle on $\Sigma$.

In the interest of counting how many basis operators generate $\fZ_\t{mag}$, it will be useful to formalize the above as follows. The dimension of the magnetic center will be the sum of the number of $p$-cycle surface operators, $\hW^{\vec w}_\eta$, and $(d-p-1)$-cycle surface operators, $\hV^{\vec v}_\sigma$, spanning $\fZ_\t{mag}$:
\begin{equation}
	\abs{\fZ_\t{mag}}=\abs{\det\bbK}^{\left(\h_\t{mag}^p+\h_\t{mag}^{d-p-1}\right)},
\end{equation}
where $\h_\t{mag}^k$ is defined in the following way.  A surface operator, \(\hat{\W}_{\eta}^{\vec{w}}\), in the magnetic center must be supported on a \(p\)-cycle that lies in the intersection of the images of the push-forward maps \(i^\aent_p\) and \(i^{\coaent}_p\). Using the push-out square
\[\begin{tikzcd}[ampersand replacement=\&]
		\pd\aent \& \aent \\
		\coaent \& \Sigma
		\arrow["j^{\pd\aent}", from=1-1, to=1-2]
		\arrow["i^\aent", from=1-2, to=2-2]
		\arrow["{i^{\coaent}}"', from=2-1, to=2-2]
		\arrow["{j^{\pd\coaent}}"', from=1-1, to=2-1]
	\end{tikzcd}\]
we see that the corresponding $p$-cycle then has to lie in the image of push-forward map \(\qty(i^\aent\circ j^{\pd\aent})_p\cong \qty(i^{\coaent}\circ j^{\pd\coaent})_p\). Utilizing the associated long-exact sequence we prove in \Cref{app:zm}
\begin{equation}\label{eq:hpmag}
	\h_\t{mag}^p=\sum_{n=0}^{p-1}(-1)^{p-1-n}\b_n(\pa\aent)+\sum_{n=0}^{p}(-1)^{p-n}\left(\b_n(\Sigma)-\dim \H_n(\Sigma,\pd\aent)\right),
\end{equation}
and similarly for $\h_\t{mag}^{d-p-1}$ via the replacement $p\rightarrow (d-p-1)$.  We remind the reader that $\b_n(\cdot)$ is the $n^\t{th}$ Betti number.  Let us point out two broad features of \Cref{eq:hpmag}.  Firstly we have the alternating sum of Betti numbers intrinsic to $\pd\aent$; as we will show later this will give contributions to the entropy analogous to those found in \cite{Grover:2011fa}.  Secondly, however, we find an interesting dependence on bulk topology relative to how $\pd\aent$ is embedded in $\Sigma$.  Although perhaps initially surprising, we can easily argue why we expect this dependence on the bulk topology to show up: $\fA_\t{mag}$ is defined with respect to homotopy equivalence within $\Sigma$. If $\h^p_\t{mag}$ only detected intrinsic topology of $\pa\aent$ it could easily\footnote{Easy examples are cooked up when $\Sigma$ is topologically trivial, e.g. a $(d-1)$-sphere, $\S^{d-1}$.} count more operators in $\fZ_\t{mag}[\aent]$ than actually exist in $\fA_\t{mag}[\aent]$, or even in $\fA[\Sigma]$! These additional bulk terms are then crucial for ensuring that this counting makes sense.

Summing $\h_\t{mag}^p$ and $\h_\t{mag}^{d-p-1}$ the total dimension of $\fZ_\t{mag}$ can be simplified utilizing the long exact sequence (see \Cref{app:zm}) to
\begin{align}\label{eq:absZmag}
	\log\abs{\fZ_\t{mag}}= & \Bigg[2\sum_{n=0}^{p-1}(-1)^{p-1-n}\b_n(\pa\aent)+\left(\b_p(\Sigma)-\dim\H_p(\Sigma,\pa\aent)\right)\nonumber      \\
	                       & \qquad\qquad+(-1)^{d-p-1}\big(\dim\H_{d-1}(\Sigma,\pa\aent)-\dim\H_0(\Sigma,\pa\aent)\big)\Bigg]\log\abs{\det\bbK}.
\end{align}
Above, all of the bulk dependence has been isolated to dimensions of $p^\t{th}$ absolute and relative homologies, plus the additional, $p$-independent term: a potential mismatch between the bottom and top relative homologies.

\subsubsection{The topological electric algebra}
In contrast with the topological magnetic algebra, whose center is generated by operators \textquote{wrapping} the entangling surface, we will pick the topological electric algebra such to be such that its center is generated by operators \textquote{piercing} the entangling surface.  Specifically, for a region, $\aent$, we define
\begin{equation}
	\fA_\t{elec}[\aent] \coloneqq \fU\set{\hat{\W}_{\eta^\sfi}^{\vec w_\sfi},\hV_{\sigma^\sfj}^{\vec v^\sfj} \suchthat \eta^\sfi,\sigma^\sfj\nothin \coaent}.
\end{equation}
In words, $\fA_\t{elec}[\aent]$ is generated by operators that cannot be deformed to being contained completely in $\coaent$.  The algebra associated to $\coaent$, the commutant of $\fA_\t{elec}[\aent]$, is generated by all operators that do not link with any cycle that cannot be deformed to be contained in $\coaent$.  We claim that this is, in fact, generated by operators that cannot be deformed to be contained in $\aent$:
\begin{equation}\label{eq:claimelec}
	\co{\left(\fA_\t{elec}[\aent]\right)}=\fU\set{\hat{\W}_{\eta^\sfi}^{\vec w_\sfi},\hV_{\sigma^\sfj}^{\vec v^\sfj}\suchthat \eta^\sfi,\sigma^\sfj\nothin\aent}=\fA_\t{elec}[\coaent].
\end{equation}
The proof of this claim is given in \Cref{app:zm}.

In this case, the center is then given by cycles of \(\Sigma\), that cannot be deformed to be contained completely in $\aent$ nor $\coaent$:
\begin{align}
	\fZ_\t{elec}[\aent] & =\fU\set{\hW^{\vec w_\sfi}_{\eta^\sfi},\hV^{\vec v_\sfj}_{\sigma^\sfj}\suchthat \eta^\sfi,\sigma^\sfj\nothin \aent\;\text{and}\;\nothin\coaent} \nn
	                    & =\fU\set{\hW^{\vec w_\sfi}_{\eta^i},\hV^{\vec v_\sfj}_{\sigma^\sfj}\suchthat \eta^\sfi\in\coker i_p^{\aent}\cap\coker i_p^{\coaent},\;\sigma^\sfj\in\coker i_{d-p-1}^{\aent}\cap\coker i_{d-p-1}^{\coaent}}.
\end{align}
$\fZ_\t{elec}[\aent]$ are topological surface operators that \emph{must} cross $\pa\aent$ non-trivially.  We name this center the \emph{topological electric center} on account of its similarity to the electric center of lattice gauge theories generated by link operators emanating transversely from the entangling surface \cite{Casini:2013rba}.  However, let us caution that this is a somewhat shallow comparison: the electric center typically discussed in lattice gauge theories is microscopic, being given by operators acting on all links intersecting $\pa\aent$ in a UV lattice realization of a topological phase.\footnote{As emphasized in \cite{Radicevic:2015sza,Lin:2018bud}, the electric center of lattice gauge theories shares many features with extending the Hilbert space with edge-mode degrees of freedom.  We discuss the extended Hilbert space of BF theory in a follow-up paper \cite{Fliss:2023uiv}.}  Our topological electric center is an extreme course-graining of this, generated by a handful of topological surface operators that are only defined up to homotopy.

With respect to counting the number of surface operators generating $\fZ_\t{elec}$,
\begin{equation}
	\abs{\fZ_\t{elec}}=\abs{\det\bbK}^{\left(\h^p_\t{elec}+\h^{d-p-1}_\t{elec}\right)},
\end{equation}
we show in appendix \Cref{app:zm} that
\begin{equation}
	\h^p_\t{elec}=\h^{d-p-1}_\t{mag},\qquad\qquad \h^{d-p-1}_\t{elec}=\h^p_\t{mag}.
\end{equation}
The heuristic argument for this follows: a $p$-cycle surface operator in $\fZ_\t{elec}$ by definition can't be deformable to either the interiors of $\aent$ or $\co\aent$ and so must wrap a basis $(d-p-1)$-cycle intrinsic to $\pa\aent$.  This cycle is precisely where one would put a $(d-p-1)$-cycle surface operator lying in $\fZ_\t{mag}$.  Consequently
\begin{equation}
	\abs{\fZ_\t{elec}}=\abs{\fZ_\t{mag}}.
\end{equation}

\subsection*{Subregion electric-magnetic duality}

The equivalence of the counting $\abs{\fZ_\t{elec}}=\abs{\fZ_\t{mag}}$ is a particular instance of a refinement of the electric-magnetic duality described in \Cref{sect:canonquant} applied to the region $\aent$ and its operator algebras.  More specifically, in \Cref{app:zm} we prove that the intersection pairing, $\bbL$, induces a one-to-one correspondence between
\begin{equation}
	\coker i_p^{\coaent}\leftrightarrow \im i_{d-p-1}^{\aent}\qq{as well as}\coker i_{d-p-1}^{\coaent}\leftrightarrow \im i_p^{\aent}.
\end{equation}
This then implies a one-to-one correspondence between operators generating $\fA_\t{elec}[\aent]$ and $\fA_\t{mag}[\aent]$.  We refer to this correspondence as \emph{subregion electric-magnetic duality.}

\subsection{Decomposing the Hilbert space}\label{sect:decompHS}

The existence of a center prohibits the tensor factorization of the global Hilbert space, $\mc H_\Sigma$, into Hilbert spaces corresponding to $\aent$ and $\co\aent$ in the following way. We will illustrate this first using the topological magnetic algebra. Currently we are organizing $\mc H_\Sigma$ by the eigenvectors of $\hW^{\fw}$, $\ket{\fv}$, which are created by acting $\hV^{\fv}$ on the condensate.  Given this, we ask: \textquote{Can we partition $\ket{\fv}$ into the eigenvalues of $\hW^{\fw}\in\fA_\t{mag}[\aent]$ and the eigenvalues of $\hW^{\fw}\in\fA_\t{mag}[{\co\aent}]$?}
\begin{equation}
	\ket{\fv}\overset{?}{=}\ket{\set{\fv^\aent},\set{\fv^{\co\aent}}}.
\end{equation}
One obvious obstruction to the above is the possible existence of $\hW^{\fw}\in\fZ_\t{mag}[\aent]$\linebreak$=\fA_\t{mag}[\aent]\cap\fA_\t{mag}[{\co\aent}]$ whose eigenvalues are overcounted in the above partition.  A more subtle obstruction to the above comes from $\hat V^{\fv}\in\fZ_\t{mag}$, which, being deformable to either inside $\aent$ or $\co\aent$, make it ambiguous if their action should shift the $\set{\fv^\aent}$ or the $\set{\fv^{\co\aent}}$ sets of eigenvalues.  To that end let us define the set $\set{\fv^{\inr{\aent}}}$ as the eigenvalues of \linebreak $\set{\hat{\W}_{\eta^\sfi}^{\vec w_\sfi}\suchthat \eta^\sfi\hin \aent\;\t{and}\;\eta^\sfi\nothin \coaent}$ and similarly $\set{\fv^{\inr{\coaent}}}$ the eigenvalues of $\set{\hat{\W}_{\eta^\sfi}^{\vec w_\sfi}\suchthat \eta^\sfi\hin \coaent\;\t{and}\;\eta^\sfi\nothin \aent}$.  We can label the eigenvalues of $\hW^{\fw}\in\fZ_\t{mag}[\aent]$ as $\set{\fv^{\pa\aent}_{\perp}}$.  The \textquote{perpendicular} notation here denotes that because they are measured by $\hW^{\fw}$ operators \textquote{living in $\pa\aent$} they are created by the action of $\hV^{\fv}$ operators which cross the entangling surface transversally. These $\hV^{\fv}$ operators do not belong in either $\fA_\t{mag}[\aent]$ or $\fA_\t{mag}[{\co\aent}]$ (in fact they are $(d-p-1)$-cycle surface operators generating $\fZ_\t{elec}$).  Lastly we will denote by $\set{\fv^{\pa\aent}_{\parallel}}$ the labels for states created by the action of $\hV^{\fv}\in\fZ_\t{mag}[\aent]$ on the condensate.  These states are eigenvectors of $\hW^{\fw}$ which also cross the entangling surface transversally (which belong neither in $\fA_\t{mag}[{\aent}]$ nor $\fA_\t{mag}[{\co\aent}]$).

Thus a general state can be partitioned unambiguously as
\begin{equation}\label{eq:MCstatepartition1}
	\ket{\fv}=\ket{\set{\fv^{\inr{\aent}}},\set{\fv^{\pa\aent}_{\perp},\fv^{\pa\aent}_{\parallel}},\set{\fv^{\inr{\coaent}}}},
\end{equation}
which have natural action by operators in the center.  I.e for all $\hW^{\fw}\in\fZ_\t{mag}[\aent]$:
\begin{equation}
	\hW^{\fw}\ket{\set{\fv^{\inr{\aent}}},\set{\fv^{\pa\aent}_{\perp},\fv^{\pa\aent}_{\parallel}},\set{\fv^{\inr{\coaent}}}}=\ex{\ii\Gamma\qty(\fv^{\pa\aent}_{\perp},\fw)}\ket{\set{\fv^{\inr{\aent}}},\set{\fv^{\pa\aent}_{\perp},\fv^{\pa\aent}_{\parallel}},\set{\fv^{\inr{\coaent}}}},
\end{equation}
while for all $\hat V^{\fv'}\in\fZ_\t{mag}[\aent]$,
\begin{equation}
	\hV^{\fv'}\ket{\set{\fv^{\inr{\aent}}},\set{\fv^{\pa\aent}_{\perp},\fv^{\pa\aent}_{\parallel}},\set{\fv^{\inr{\coaent}}}}=\ket{\set{\fv^{\inr{\aent}}},\set{\fv^{\pa\aent}_{\perp},\fv^{\pa\aent}_{\parallel}+\fv'},\set{\fv^{\inr{\coaent}}}}.
\end{equation}
Under this partitioning the Hilbert space decomposes as
\begin{equation}\label{eq:Hdecomp1}
	\mc H_\Sigma=\bigoplus_{\set{\fv^{\pa\aent}}}\mc H_\Sigma^{\left(\fv^{\pa\aent}\right)},
\end{equation}
where we've used a short-hand, $\set{\fv^{\pa\aent}}\coloneqq\set{\fv^{\pa\aent}_{\perp},\fv^{\pa\aent}_{\parallel}}$. Each block in this decomposition admits a tensor product on $\inr{\aent}$ and $\inr{\coaent}$:
\begin{equation}
	\mc H^{\left(\fv^{\pa\aent}\right)}_\Sigma=\mc H^{\left(\fv^{\pa\aent}\right)}_{\inr{\aent}}\otimes \mc H^{\left(\fv^{\pa\aent}\right)}_{\inr{\coaent}}.
\end{equation}
This partitioning is useful to illustrate the obstruction of a global tensor-product decomposition of $\mc H_\Sigma$, however it is not computationally useful since the action of $\hV^{\fv}\in\fZ_\t{mag}[\aent]$ moves between different blocks of \Cref{eq:Hdecomp1}.  It is more helpful to instead diagonalize their action.  Their eigenvectors are associated with $\hW^{\fw}$ operators crossing the entangling surface transversally and which $\hV^{\fv}\in\fZ_\t{mag}[\aent]$ link non-trivially: these are the $p$-cycle surface operators generating $\fZ_\t{elec}$.  Because of this we will label the eigenvectors with the set $\set{\fw^{\pa\aent}_{\perp}}$ and partition our system as
\begin{equation}\label{eq:magstatepart}
	\ket{\set{\fv^{\inr{\aent}}},\set{\fv^{\pa\aent}_{\perp},\fw^{\pa\aent}_{\perp}},\set{\fv^{\inr{\coaent}}}},
\end{equation}
such that for all $\hV^{\fv}\in\fZ_\t{mag}[\aent]$
\begin{equation}
	\hV^{\fv}\ket{\set{\fv^{\inr{\aent}}},\set{\fv^{\pa\aent}_{\perp},\fw^{\pa\aent}_{\perp}},\set{\fv^{\inr{\coaent}}}}=\ex{\ii\Gamma\qty(\fv,\fw^{\pa\aent}_{\perp})}\ket{\set{\fv^{\inr{\aent}}},\set{\fv^{\pa\aent}_{\perp},\fw^{\pa\aent}_{\perp}},\set{\fv^{\inr{\coaent}}}}.
\end{equation}
This amounts to the change of basis
\begin{equation}\ket{\set{\fv^{\inr{\aent}}},\set{\fv^{\pa\aent}_{\perp},\fv^{\pa\aent}_{\parallel}},\set{\fv^{\inr{\coaent}}}}=\sum_{\fw^{\pa\aent}_{\perp}}\ex{\ii\Gamma\qty(\fv^{\pa\aent}_{\parallel},\fw^{\pa\aent}_{\perp})}\ket{\set{\fv^{\inr{\aent}}},\set{\fv^{\pa\aent}_{\perp},\fw^{\pa\aent}_{\perp}},\set{\fv^{\inr{\coaent}}}}
\end{equation}
on the $\set{\fv^{\pa\aent}_\parallel}$ part of the state.

With respect to this partitioning, we can decompose the global Hilbert space with respect to the eigenvalues of operators spanning $\fZ_\t{mag}[\aent]$:
\begin{equation}\label{eq:Hdecomp2}
	\mc H_\Sigma=\bigoplus_{\set{\mf q^{\pa\aent}_{\perp}}}\mc H_\Sigma^{\left(\mf q^{\pa\aent}_{\perp}\right)},
\end{equation}
where we now use the short-hand, $\set{\mf q^{\pa\aent}_{\perp}}\coloneqq\set{\fv^{\pa\aent}_{\perp},\fw^{\pa\aent}_{\perp}}$.  Again, each block admits a tensor product on $\inr{\aent}$ and $\inr{\coaent}$:
\begin{equation}\label{eq:Htensorprod2}
	\mc H_\Sigma^{\left(\mf q^{\pa\aent}_{\perp}\right)}=\mc H^{\left(\mf q^{\pa\aent}_{\perp}\right)}_{\inr{\aent}}\otimes \mc H^{\left(\mf q^{\pa\aent}_{\perp}\right)}_{\inr{\coaent}}.
\end{equation}
In this set-up it is easy to see from our state partitions, \Cref{eq:magstatepart}, that all $\mc H^{\left(\mf q^{\pa\aent}_{\perp}\right)}_{\inr{\aent}}\otimes \mc H^{\left(\mf q^{\pa\aent}_{\perp}\right)}_{\inr{\coaent}}$ blocks are isomorphic.  To make the following discussion notationally cleaner will often drop the $(\fq_\perp^{\pa\aent})$ superscripts from the tensor factors $\mc H_{\inr\aent}$ and $\mc H_{\inr\coaent}$.

For the topological electric algebra, we can decompose the Hilbert space based on the eigenvalues of the operators generating the center in a similar way.  Again, the partitioning of cycles into those that can and cannot be deformed being contained in $\pa\aent$ provides a useful partition of surface operator charges.  It is precisely the operators contained in $\pa\aent$ (i.e. those in $\fZ_\t{mag}[\aent]$) that can link non-trivally with operators in $\fZ_\t{elec}[\aent]$ and so their charges, $\set{\fv^{\pa R}_{\parallel},\fw^{\pa R}_{\parallel}}$ are the eigenvalues of the basis generating $\fZ_\t{elec}[\aent]$.  A generic state of $\mc H_\Sigma$ then can be partitioned as
\begin{equation}\label{eq:tEstatesplit}
	\ket{\set{\fv^{\inr{\aent}}},\set{\fv^{\pa R}_{\parallel},\fw^{\pa R}_{\parallel}},\set{\fv^{\inr{\coaent}}}},
\end{equation}
corresponding to a Hilbert space decomposition
\begin{equation}\label{eq:Hdecomp3}
	\mc H_\Sigma=\bigoplus_{\set{\mf q^{\pa \aent}_{\parallel}}}\mc H_\Sigma^{\left(\mf q^{\pa\aent}_{\parallel}\right)}~,
\end{equation}
where $\set{\mf q^{\pa \aent}_{\parallel}}\coloneqq\set{\fv^{\pa\aent}_{\parallel},\fw^{\pa\aent}_{\parallel}}$ is a short-hand notation for the central eigenvalues.  Each block of \Cref{eq:Hdecomp3} admits a tensor product on $\inr{\aent}$ and $\inr{\coaent}$:
\begin{equation}\label{eq:Htensorprod3}
	\mc H_\Sigma^{\left(\mf q^{\pa\aent}_{\parallel}\right)}=\mc H_{\inr{\aent}}^{\left(\mf q^{\pa\aent}_{\parallel}\right)}\otimes\mc H_{\inr{\coaent}}^{\left(\mf q^{\pa\aent}_{\parallel}\right)}.
\end{equation}
Again, since these blocks are all isomorphic, in what follows we will omit the $(\fq_\parallel^{\pa\aent})$ superscripts from the tensor factors unless clarity demands it.

\subsection{Decomposing a state and the reduced density matrix}\label{sect:redDM}

Now let us consider a generic state, $\rho$, on $\mc H_\Sigma$.  With respect to either Hilbert space decomposition, \Cref{eq:Hdecomp2} or \Cref{eq:Hdecomp3},
\begin{equation}
	\mc H_\Sigma=\bigoplus_{\set{\fq}}\mc H_\Sigma^{\left(\fq\right)}=\bigoplus_{\set{\fq}}\mc H_{\inr\aent}^{\left(\fq\right)}\otimes\mc H_{\inr\coaent}^{\left(\fq\right)},\qquad \fq\in\set{\fq_\perp^{\pa\aent},\fq_\parallel^{\pa\aent}},
\end{equation}
we can write
\begin{equation}\label{eq:rhodecomp}
	\rho=\bigoplus_{\set{\fq}}\lambda_{\left(\fq\right)}\rho^{\left(\fq\right)},\qquad\qquad \fq\in\set{\fq^{\pa\aent}_\perp,\fq^{\pa\aent}_\parallel},
\end{equation}
which corresponds to diagonalizing $\rho$ as an operator with respect to the respective center, $\fZ[\aent]$.  The coefficients, $\lambda_{\left(\fq\right)}$, are chosen such that $\rho^{\left(\fq\right)}$ are normalized states on $\mc H_\Sigma^{\left(\fq\right)}$:
\begin{equation}
	\Tr_{\mc H_\Sigma^{\left(\fq\right)}}\rho^{\left(\fq\right)}=1,
\end{equation}
which requires
\begin{equation}
	\sum_{\set{\fq}}\lambda_{\left(\fq\right)}=1.
\end{equation}
Within each block $\mc H_\Sigma^{\left(\fq\right)}$ we can construct the reduced density matrix with respect to the tensor product, \Cref{eq:Htensorprod2} or \Cref{eq:Htensorprod3}, by tracing out $\mc H_{\inr{\coaent}}$:
\begin{equation}\label{eq:blockredrho}
	\rho^{\left(\fq\right)}_{\inr{\aent}}=\Tr_{\mc H_{\inr{\coaent}}}\rho^{\left(\fq\right)}.
\end{equation}
This results in a reduced density matrix on $\aent$ which follows from the sum decomposition, \Cref{eq:rhodecomp}:
\begin{equation}\label{eq:redrhodecomp}
	\rho_\aent=\bigoplus_{\set{\fq}}\lambda_{\left(\fq\right)}\left(\rho^{\left(\fq\right)}_{\inr{\aent}}\otimes\tau_{\mc H_{\inr{\coaent}}}\right).
\end{equation}
where $\tau_{\mc H}$ is the trace-normalized identity matrix on a Hilbert space:
\begin{equation}
	\tau_{\mc H}\coloneqq \frac{\hat 1_{\mc H}}{\dim\mc H}.
\end{equation}
Alternatively, we can define this reduced density matrix as the unique unit-trace, Hermitian operator in $\fA[\aent]$ which reproduces expectation values of all other operators in $\fA[\aent]$:
\begin{equation}\label{eq:redrho_algdef}
	\rho_\aent\in\fA[\aent]\qquad\text{such that}\qquad\Tr_{\mc H_\Sigma}\left(\rho_\aent\mc O_\aent\right)=\Tr_{\mc H_\Sigma}\left(\rho\mc O_\aent\right)\qquad\fall\;\mc O_\aent\in\fA[\aent].
\end{equation}
Given the decomposition \Cref{eq:redrhodecomp}, we assign an entanglement entropy to $\rho$ and $\fA[\aent]$ via
\begin{equation}\label{eq:algEEdef}
	\alge{\fA[\aent]}[\rho]\coloneqq\sum_{\set{\fq}}\lambda_{\left(\fq\right)}\vne\left[\rho_{\inr{\aent}}^{\left(\fq\right)}\right]-\sum_{\set{\fq}}\lambda_{\left(\fq\right)}\log\lambda_{\left(\fq\right)},
\end{equation}
where
\begin{equation}\label{eq:vNentdef}
	\vne\left[\rho_{\inr{\aent}}^{(\fq)}\right]\coloneqq-\Tr_{\mc H_{\inr{\aent}}}\left(\rho_{\inr{\aent}}^{(\fq)}\log\rho_{\inr{\aent}}^{(\fq)}\right)
\end{equation}
is the von Neumann entropy of the reduced density matrix in a fixed block,  \Cref{eq:blockredrho}.  The second term of \Cref{eq:algEEdef} is a Shannon entropy of the probability of measuring the set $\set{\fq}$ of eigenvalues in the state $\rho$.  We emphasize that, while ostensibly classical in nature, this Shannon entropy is present even for pure quantum states and is a generic feature of entanglement entropies associated to centerful algebras.  It will play an important role in the following section.  For a centerless algebra (such as $\fA_\t{aus}[\aent]$ or $\fA_\t{greedy}[\aent]$), the entanglement entropy, $\alge{\fA[\aent]}[\rho]$, can be defined in the usual way as the von Neumann entropy of $\rho$ reduced on the corresponding tensor-factorization.

\subsection{Essential topological entanglement}
We now illustrate the entanglement entropy associated to various choices of subregion algebra for a ground state, $\ket{\psi}=|\fv_\star\rangle$, for some fixed $\fv_\star$.  As discussed in \Cref{sect:canonquant}, this state is prepared by the action of $\hV^{\fv_\star}$ on a condensate. In the context of the BF path-integral, these states are natural: they are prepared by path-integral on the interior of a manifold having $\Sigma$ as its boundary and with $\hV$ Wilson surface operators inserted.  In $(2+1)$ dimensional topological phases, such states play a key role in the standard treatment of entanglement entropy (say via the replica trick, lattice regularization, or an extended Hilbert space), yielding the celebrated \textquote{topological entanglement entropy} discussed in \Cref{sect:intro}.  The primary upshot of this section is to show that in our setup such states provide a new smoking gun topological signature in the algebraic entropy, which we name the \textquote{essential topological entanglement.}  There will be two varieties: one associated to the topological magnetic algebra, and one associated to the topological electric algebra:
\begin{equation}
	\cE_{\t{mag/elec}}\coloneqq\alge{\fA_\t{mag/elec}[\aent]}\qty[|\fv_\star\rangle\langle\fv_\star|].
\end{equation}
We will further see that these two essential topological entanglements are related by a electric-magnetic duality.

Let us begin the discussion with the topological magnetic algebra in mind.  We will illustrate the machinery in this instance; the consideration of the other subregion algebras will be wholly clear afterward.  We will consider the global pure state
\begin{equation}\label{eq:state1}
	\rho=\sum_{\fv_1,\fv_2}\psi_{\fv_1}\psi^\ast_{\fv_2}\,\kb{\fv_1}{\fv_2}.
\end{equation}
and reduce it down on $\fA_\t{mag}[\aent]$ for some region $\aent$.  We can express the reduced density matrix of the state, \Cref{eq:state1}, as element of $\fA_\t{mag}[\aent]$ via
\begin{equation}\label{eq:purerdDM}
	\rho_R=\mc N^{-1}\sum_{\fw^\aent}\sum_{\fv^\aent}\sum_{\fv_1}\psi_{\fv_1}\psi^\ast_{\fv_1-\fv^\aent}\ex{-\ii\Gamma\qty(\fv_1,\fw^\aent)}\,\hW^{\fw^\aent}\,\hV^{\fv^\aent}.
\end{equation}
where $\mc N=\dim\mc H_\Sigma=\abs{\det\bbK}^{\b_p(\Sigma)}$ ensures that $\rho_\aent$ is trace normalized over $\mc H_\Sigma$.  See \Cref{app:redDMdecomp} for details on this decomposition.  It will be useful to split this normalization into
\begin{equation}
	\mc N=\mc N_{\inr\aent}\,\mc N_{\pa\aent}\,\mc N_{\inr\coaent}=\abs{\det\bbK}^{\h_{\inr\aent}}\,\abs{\det\bbK}^{\h_{\pa\aent}}\,\abs{\det\bbK}^{\h_{\inr\coaent}},
\end{equation}
where $\h_{\inr\aent}$ is the number of independent $(d-p-1)$-cycles deformable to $\aent$ but not to $\coaent$ (i.e. the number of $(d-p-1)$-cycles spanning the austere algebra for $\aent$).  And vice-versa for $\h_{\inr\coaent}$.  Thus $\mc N_{\inr\aent}$ and $\mc N_{\inr\coaent}$ are simply the dimensions of $\mc H_{\inr\aent}$ and $\mc H_{\inr\coaent}$, respectively, in the decomposition \Cref{eq:Htensorprod2}. Above, $\h_{\pa\aent}$ is the number of $(d-p-1)$-cycles piercing $\pa\aent$, $\h^{d-p-1}_\t{elec}$, (counting the independent $\set{\fv_\perp^{\pa\aent}}$ in the decomposition \Cref{eq:magstatepart}) plus $\h_\t{mag}^{d-p-1}$, the number of $(d-p-1)$-cycles deformable to $\pa\aent$.  As discussed above in \Cref{sect:decompHS}, the latter of these is equal to $\h^p_\t{elec}$, and counts the independent $\set{\fw_\perp^{\pa\aent}}$ in \Cref{eq:magstatepart}.  Thus
\begin{equation}
	\h^{\pa R}=\h^p_\t{elec}+\h^{d-p-1}_\t{elec}=\h^p_\t{mag}+\h^{d-p-1}_\t{mag}.
\end{equation}

In writing \Cref{eq:purerdDM}, some of the $\hW$'s and $\hV$'s belong to the center, $\fZ_\t{mag}[\aent]$.  It will be useful to separate them off as
\begin{align}
	\hW^{\fw^\aent}= & \prod_{\eta^\sfj\hin\aent,\nothin\coaent}\hW_{\eta^\sfj}^{\vec w^\aent_\sfj}\prod_{\eta^\sfj\hin\aent,\hin\coaent}\hW_{\eta^\sfj}^{\vec w^\aent_\sfj}\equiv \hW^{\fw^{\inr{\aent}}}\hW^{\fw^{\pa\aent}_\parallel}\nonumber \\
	\hV^{\fv^\aent}= & \prod_{\sigma^\sfi\hin\aent,\nothin\coaent}\hV_{\sigma^\sfi}^{\vec v^\aent_\sfi}\prod_{\sigma^\sfi\hin\aent,\hin\coaent}\hV_{\sigma^\sfi}^{\vec v^\aent_\sfi}\equiv \hV^{\fv^{\inr{\aent}}}\hV^{\fv^{\pa\aent}_\parallel}.
\end{align}
Diagonalizing these central elements, we can write them in terms of their eigenvalues as
\begin{equation}
	\hW^{\fw^{\pa\aent}_\parallel}=\bigoplus_{\fv^{\pa\aent}_\perp}\ex{\ii\Gamma\qty(\fv_\perp^{\pa\aent},\fw^{\pa\aent}_\parallel)},\qquad\qquad\hV^{\fv^{\pa\aent}_\parallel}=\bigoplus_{\fw^{\pa\aent}_\perp}\ex{\ii\Gamma\qty(\fv_\parallel^{\pa\aent},\fw^{\pa\aent}_\perp)}.
\end{equation}
This leads to a reduced density matrix in block diagonal form, as in \Cref{eq:redrhodecomp}:
\begin{equation}\label{eq:magredrho}
	\rho_R=\bigoplus_{\set{\fv^{\pa\aent}_\perp,\fw^{\pa\aent}_\perp}}\lambda_{(\fv^{\pa\aent}_\perp,\fw^{\pa\aent}_\perp)}\left(\rho^{(\fv^{\pa\aent}_\perp,\fw^{\pa\aent}_\perp)}_{\inr{\aent}}\otimes \tau_{\mc H_{\inr{\coaent}}}\right),
\end{equation}
with
\begin{align}
	\lambda_{(\fv^{\pa\aent}_\perp,\fw^{\pa\aent}_\perp)}=    & \mc N_{\pa\aent}^{-1}\sum_{\fw^{\pa\aent}_\parallel}\sum_{\fv^{\pa\aent}_\parallel}\sum_{\fv_1}\psi^\ast_{\fv_1-\fv^{\pa\aent}_\parallel}\psi_{\fv_1}\ex{-\ii\Gamma\qty(\fv_1-\fv^{\pa\aent}_\perp,\fw^{\pa\aent}_\parallel)+\ii\Gamma\qty(\fv^{\pa\aent}_\parallel,\fw^{\pa\aent}_\perp)}\nonumber \\
	\rho_\aent^{(\fv^{\pa\aent}_\perp,\fw^{\pa\aent}_\perp)}= & \mc N_{\inr\aent}^{-1}\sum_{\fw^{\inr{\aent}}}\sum_{\fv^{\inr{\aent}}}\sum_{\fv_1}\psi^\ast_{\fv_1-\fv^{\inr{\aent}}}\psi_{\fv_1}\ex{-\ii\Gamma\qty(\fv_1,\fw^{\inr{\aent}})}\hW^{\fw^{\inr{\aent}}}\,\hV^{\fv^{\inr{\aent}}},
\end{align}
and $\tau_{\mc H_{\inr{\coaent}}}$ is the trace-normalized unit operator of $\mc H_{\inr{\coaent}}$ with respect to the decomposition \Cref{eq:Htensorprod2}.  We can easily verify that
\begin{equation}
	\sum_{\fv^{\pa\aent}_\perp}\sum_{\fw^{\pa\aent}_\perp}\lambda_{(\fv_\perp^{\pa\aent},\fw^{\pa\aent}_\perp)}=\sum_{\fv_1}\psi^\ast_{\fv_1}\psi_{\fv_1}=1,
\end{equation}
because the sums over $\fw^{\pa\aent}_\perp$ and $\fv^{\pa\aent}_\perp$ enforce delta functions on $\fv^{\pa\aent}_\parallel$ and $\fw^{\pa\aent}_\parallel$, respectively.

We now will take our pure state to be a ground state, $|\psi\rangle=|\fv_\star\rangle$ for fixed $\fv_\star$, which sets $\psi_\fv={\boldsymbol\delta}_{\fv-\fv_\star}$ (defined in \Cref{eq:latticedelta}).  With respect to the decomposition \Cref{eq:Htensorprod2} $|\fv_\star\rangle$ projected onto each central eigenspace remains a product state on $\mc H_{\inr\aent}\otimes\mc H_{\inr\coaent}$ and so $\rho_{\inr\aent}^{(\fv_\perp^{\pa\aent},\fw_\perp^{\pa\aent})}$ is pure on $\mc H_{\inr\aent}$:
\begin{equation}
	\rho_{\inr\aent}^{(\fv_\perp^{\pa\aent},\fw_\perp^{\pa\aent})}=\mc N_{\inr\aent}^{-1}\sum_{\fw^{\inr\aent}}\ex{-\ii\Gamma\qty(\fv_\star,\fw^{\inr\aent})}\,\hW^{\fw^{\inr\aent}}=\kb{\fv^{\inr\aent}_\star}{\fv^{\inr\aent}_\star}.
\end{equation}
As such its von Neumann entropy vanishes, $\vne\left[\rho_{\inr\aent}^{\left(\fq_\perp^{\pa\aent}\right)}\right]=0$. The entanglement entropy of $|\fv_\star\rangle$ then comes entirely from the Shannon entropy of the distribution, $\set{\lambda_{\left(\fq_\perp^{\pa\aent}\right)}}$, which take the form
\begin{equation}
	\lambda_{(\fv_\perp^{\pa\aent},\fw_\perp^{\pa\aent})}=\mc N^{-1}_{\pa\aent}\sum_{\fw_\parallel^{\pa\aent}}\ex{-\ii\Gamma\qty(\fv_\star-\fv_\perp^{\pa\aent},\fw_\parallel^{\pa\aent})}={\boldsymbol\delta}_{\fv_{\star\perp}^{\pa\aent}-\fv_\perp^{\pa\aent}}\abs{\det\bbK}^{-\h_\t{mag}^{d-p-1}}~;
\end{equation}
that is, their support is isolated to the $\fv_\perp^{\pa\aent}$ sector determined by the original ground state, and is maximally mixed on the $\fw_\perp^{\pa\aent}$ eigenvalues.  The algebraic entanglement entropy associated to $\fA_\t{mag}[\aent]$ of $|\fv_\star\rangle$ then determines the essential topological entanglement as
\begin{empheq}[box=\obox]{equation}\label{eq:magEEbasis}
	\cE_\t{mag}=-\sum_{\set{\fv^{\pa\aent}_\perp,\fw^{\pa\aent}_\perp}}\lambda_{\left(\fv^{\pa\aent}_\perp,\fw^{\pa\aent}_\perp\right)}\log\lambda_{\left(\fv^{\pa\aent}_\perp,\fw^{\pa\aent}_\perp\right)}=\h^{d-p-1}_\t{mag}\log\abs{\det\bbK},
\end{empheq}
where we remind the reader
\begin{equation}
	\h_\t{mag}^{d-p-1}=\sum_{n=0}^{d-p-2}(-1)^{d-p-2-n}\b_n(\pa\aent)+\sum_{n=0}^{d-p-1}(-1)^{d-p-1-n}\left(\b_n(\Sigma)-\dim \H_n(\Sigma,\pd\aent)\right).
\end{equation}
Let us dissect this result. We make note of several features of \Cref{eq:magEEbasis}.
\begin{itemize}
	\item $\cE_\t{mag}$ probes non-local operator statistics through $\abs{\det\bbK}$ which can be regarded as a \textquote{total quantum dimension} in this theory.
	\item $\cE_\t{mag}$ is independent of the ground state, $\ket{\fv_\star}$, in which it is evaluated. This is in keeping with this being an Abelian topological phase.  Since operator fusion is unique in $\fA[\Sigma]$, the \textquote{quantum dimension,} $\cD_{\fv_\star}$, of the $\hV^{\fv_\star}$ building $\ket{\fv_\star}$ from the condensate is unity and so any possible contribution going as $\log\cD_{\fv_\star}$ will vanish.
	\item $\cE_\t{mag}$ probes topological features intrinsic to $\pa\aent$: the alternating sum of Betti numbers on $\pa\aent$.  This contribution mimics a proposed (negative) correction to the area law in higher-dimensional topological order described by membrane-net models \cite{Grover:2011fa}. Here we find this term contributes positively and appears without an area law.
	\item $\cE_\t{mag}$ possesses additional terms that depend both on the topology of $\Sigma$ itself, as well as the relative homologies of $\Sigma$ and $\pa\aent$.  Thus, the essential topological entanglement is sensitive to more than the topology of $\pa\aent$ itself, but also how $\pa\aent$ is embedded into $\Sigma$.\footnote{Note, however, that $\cE_\t{mag}$ is insensitive to Euler-character ambiguities, identified in \cite{Grover:2011fa}  as both local and global. See, for instance, \cref{eq:hmag-sum-app} showcasing this fact in an example calculation.}  As we argued in \Cref{sect:topsubalgs}, this has to be the case since the operators counted by $\cE_\t{mag}$ have to descend from non-trivial topological operators on $\Sigma$.
	\item $\cE_\t{mag}$ comes entirely from the Shannon contribution to $\alge{\fA_\t{mag}}[{\kb{\fv_\star}{\fv_\star}}]$, while the contribution from the sum of von Neumann entropies exactly vanish. It was argued that this latter contribution corresponds to the distillable entanglement in gauge theories \cite{Soni:2015yga,VanAcoleyen:2015ccp}. That $\cE_\t{mag}$ entirely enters through the Shannon term is in keeping with its essential non-locality: it cannot be distilled into Bell pairs by local operations.
\end{itemize}
We can repeat this same exercise for the entanglement entropy associated to $\fA_\t{elec}[\aent]$. Running through this process one finds for a pure state the reduced density matrix decomposes in a way wholly similar to \Cref{eq:magredrho}:
\begin{equation}
	\rho_R=\bigoplus_{\set{\fv^{\pa\aent}_\parallel,\fw^{\pa\aent}_\parallel}}\lambda_{(\fv^{\pa\aent}_\parallel,\fw^{\pa\aent}_\parallel)}\left(\rho^{(\fv^{\pa\aent}_\parallel,\fw^{\pa\aent}_\parallel)}_{\inr{\aent}}\otimes \tau_{\mc H_{\inr{\coaent}}}\right),
\end{equation}
with
\begin{align}
	\lambda_{(\fv^{\pa\aent}_\parallel,\fw^{\pa\aent}_\parallel)}=    & \mc N_{\pa\aent}^{-1}\sum_{\fw^{\pa\aent}_\perp}\sum_{\fv^{\pa\aent}_\perp}\sum_{\fv_1}\psi^\ast_{\fv_1-\fv^{\pa\aent}_\perp}\psi_{\fv_1}\ex{-\ii\Gamma\qty(\fv_1-\fv^{\pa\aent}_\parallel,\fw^{\pa\aent}_\perp)+\ii\Gamma\qty(\fv^{\pa\aent}_\perp,\fw^{\pa\aent}_\parallel)}\nonumber \\
	\rho_\aent^{(\fv^{\pa\aent}_\parallel,\fw^{\pa\aent}_\parallel)}= & \mc N_{\inr\aent}^{-1}\sum_{\fw^{\inr{\aent}}}\sum_{\fv^{\inr{\aent}}}\sum_{\fv_1}\psi^\ast_{\fv_1-\fv^{\inr{\aent}}}\psi_{\fv_1}\ex{-\ii\Gamma\qty(\fv_1,\fw^{\inr{\aent}})}\hW^{\fw^{\inr{\aent}}}\,\hV^{\fv^{\inr{\aent}}}.
\end{align}
Again, for a fixed ground state, $\ket{\psi}=\ket{\fv_\star}$, the reduced density matrix projected to a fixed block of \Cref{eq:Hdecomp3} is pure and the electric entanglement entropy comes entirely from the Shannon entropy of $\set{\lambda_{(\fq_\parallel^{\pa\aent})}}$.  This distribution again is isolated onto a specific block of $\fv_\parallel^{\pa\aent}$ eigenvalues determined by the ground state and are maximally mixed onto the $\fw_\parallel^{\pa\aent}$ eigenvalues:
\begin{equation}
	\lambda_{(\fv_\parallel^{\pa\aent},\fw_\parallel^{\pa\aent})}=\mc N^{-1}_{\pa\aent}\sum_{\fw_\perp^{\pa\aent}}\ex{-\ii\Gamma\qty(\fv_\star-\fv_\parallel^{\pa\aent},\fw_\perp^{\pa\aent})}={\boldsymbol\delta}_{\fv_{\star\parallel}^{\pa\aent}-\fv_\parallel^{\pa\aent}}\abs{\det\bbK}^{-\h_\t{elec}^{d-p-1}}.
\end{equation}
$\cE_\t{elec}$, being given by the algebraic entanglement entropy associated to $\fA_\t{elec}[\aent]$ of a ground state, is then
\begin{empheq}[box=\obox]{equation}
	\cE_\t{elec}=\h_\t{elec}^{d-p-1}\log\abs{\det\bbK},
\end{empheq}
where we remind the reader
\begin{equation}
	\h_\t{elec}^{d-p-1}=\h_\t{mag}^p=\sum_{n=0}^{p-1}(-1)^{p-1-n}\b_n(\pa\aent)+\sum_{n=0}^{p}(-1)^{p-n}\left(\b_n(\Sigma)-\dim \H_n(\Sigma,\pd\aent)\right).
\end{equation}
Again we see many familiar features of this essential topological entanglement: the dependence of $\log\abs{\det\bbK}$ with a coefficient displaying topological dependence of $\pa\aent$ as well as $\Sigma$ and how $\pa\aent$ is embedded into $\Sigma$.  Comparing with \Cref{eq:magEEbasis} we also notice
\begin{equation}
	\cE_\t{mag} \xmapsto{p\mapsto (d-p-1)} \cE_\t{elec},
\end{equation}
which ultimately stems from the electric-magnetic duality discussed in \Cref{sect:canonquant}.

Lastly for sake of completeness, we consider the centerless algebras, $\fA_\t{aus}$ and $\fA_\t{greedy}$. Since the ground states are already product states on the tensor factorizations defined by either $\fA_\t{aus}[\aent]$ or $\fA_\t{greedy}[\aent]$, these two algebras yield zero entanglement entropy:
\begin{equation}
	\alge{\fA_\t{aus}[\aent]}\left[\kb{\fv_\star}{\fv_\star}\right]=\alge{\fA_\t{greedy}[\aent]}\left[\kb{\fv_\star}{\fv_\star}\right]=0.
\end{equation}

\section{Discussion}\label{sect:disc}

In this paper we considered the algebraic approach to entanglement entropy applied to the algebra of surface operators in Abelian BF theories. On a technical level the algebraic approach allowed us to address issues of Hilbert space factorization while respecting the non-local and topological nature of the gauge-invariant observables available in the low-energy theory. On a conceptual level this investigation was predicated on finding a suitable definition of topological entanglement that (i) can be defined intrinsically in the IR TQFT (i.e. without the need to embed into a microscopic model, or to extend the Hilbert space with \textquote{edge modes}) and (ii) is sensitive to longer-range and intricate forms of topology than that of the entangling surface itself. To that end we defined two non-trivial algebras that can be assigned to a subregion: the topological magnetic and electric algebras. To each of these algebras we associated an algebraic entanglement entropy which we coin the \textquote{essential topological entanglement.}  Our essential topological entanglement is manifestly finite, positive, and displays a more intricate and long-range features than the topology of the entangling surface itself: namely, how the entangling surface is embedded into the Cauchy slice.

Let us comment on some open questions and open directions implied by this research below.

\subsubsection*{Comparing with traditional TEE}

The essential topological entanglement shares familiar features with traditional TEE: e.g. the log dependence on the total quantum dimension, $\abs{\det\bbK}$, and the appearance of the alternating sum of $\b_k(\pa\aent)$ which has been argued to be the coefficient of the $\log\abs{\det\bbK}$ in higher dimensions \cite{Grover:2011fa}. However, the additional dependence of $\cE$ on bulk topology makes it clear that $\cE$ truly a different object than the TEE. We have emphasized above and will emphasize again that this has to be the case.  A simple example to keep in mind when comparing the two concepts is when the region is a $D$-ball: $\aent=\bbB^D$. There is simply no non-trivial operator one can assign to either $\fA_\t{mag}$ or $\fA_\t{elec}$: physically the BF theory has integrated out all local degrees of freedom and there is no probe that can distinguish $\aent$ from the empty set. It is easy to see that $\cE_\t{mag/elec}=0$ in this case.  However, the TEE proposed by \cite{Grover:2011fa} will generically be non-zero: this is because $\pa \bbB^D=\S^{D-1}$ can support a top and bottom homology group. Again this difference stems from the fact that the TEE arises from the long-range correlations amongst UV degrees of freedom localized to $\pa\aent$ while $\cE$ arises from the long-range correlations of long-range operators delocalized on $\Sigma$.

With that difference stated, we can still speculate on the form of the traditional TEE in BF theory. As discussed above, accessing this TEE is contingent on adding in UV degrees of freedom. However, we can easily do this by extending the Hilbert space using the methods in \cite{Fliss:2017wop,Fliss:2020cos} or by regulating a replica path-integral with \textquote{edge modes.} This is the subject of immediate follow-up to this paper \cite{Fliss:2023uiv}. We will show that the inheritance of gauge transformations on an entangling surface is an infinite dimensional algebra that completely organizes the entanglement spectrum of an edge-mode theory living on $\pa\aent$. This algebra is a direct higher-dimensional analogue of the Kac--Moody algebras arising on the boundaries of Chern-Simons theories and provides a natural procedure for constructing the extended Hilbert space. While in general dimensions this algebra is not conformal, it lends many tools reminiscent of conformal field theory, e.g. a state-operator correspondence for non-local operators \cite{Hofman:2024oze}. The computation of the entanglement entropy of a subregion is entirely controlled by this algebra and leads to area and sub-area laws plus a constant correction depending on Betti numbers of $\pa\aent$.

\subsubsection*{Accessing essential topological entanglement on the lattice}

Our definition of the essential topological entanglement is strongly motivated by the IR effective TQFT described a topological phase.  In practice, however, it is much more useful to work directly with spin lattice models or with tensor network constructions of ground states.  How does one define $\cE$ in these settings?

A natural starting place are the operator algebras defined in lattice gauge theories defined on general graphs. One can then look for a projected set of gauge invariant operators that are both homotopy invariant as well as independent to refining or coarse graining the graph. Such algebras were precisely considered in \cite{Delcamp:2016eya} where an algebra of \textquote{ribbon operators} were used to define algebraic entanglement entropies in lattice gauge theories that are graph-independent, topological, and finite dimensional. These features resonate strongly with our definition of essential topological entanglement.  In that paper all entangling regions have trivial topology and so the contribution to the entanglement entropy comes entirely from surface operators terminating on (non-Abelian) quasi-particle punctures in the region.  In this work we have not considered punctured states; additionally it is unlikely they will contribute to $\cE$ because of the Abelian fusion of surface operators. This makes a direct comparison to difficult.  It would be interesting to extend the methods of \cite{Delcamp:2016eya} to more interesting topologies to investigate if our notions of algebraic entanglement coincide.

More broadly, it is fair to ask if essential topological entanglement, either defined in a TQFT or in a lattice gauge theory, affords any practical advantages over traditional TEE. A well known use for TEE is to diagnose whether a tensor product ansatz for a gapped Hamiltonian truly captures topological order \cite{Jiang:2012uea}. With this regard we do not expect $\cE$ to provide any significant advantages. However, essential topological entanglement likely displays conceptual advantages in models where manifest background independence and diffeomorphism invariance are desired, such as loop quantum gravity \cite{Delcamp:2016eya} or tensor network models of quantum gravity \cite{AkersSoni}.

\subsubsection*{Relation to multi-boundary entanglement} In this paper we have defined and evaluated the essential topological entanglement in fixed ground states. However there may be other states where the essential topological entanglement plays an important role. Notable examples are states prepared by path-integral on three-sphere link complements introduced in \cite{Balasubramanian:2016sro}. The multi-boundary entanglement across different link components escapes the standard paradigm of topological entanglement entropy precisely because each subsystem is collection of tori with no boundary. There are no edge modes and there is no area law; the topological contributions (which are related to link invariants) appear as positive and finite contributions to the entropy. The essential topological entanglement of any basis state vanishes for boundary-less regions in the $d=3$, and $p=1$ theory. However, since the Hilbert spaces appearing in \cite{Balasubramanian:2016sro} are finite dimensional and possess a basis defined by the action of line operators, it is easy to argue that the entropies appearing there can be recast as algebraic entropies of line operators. At least for the Abelian theory, the essential topological entanglement is the right arena to do so. In these cases all of the essential entanglement appears from the mixture of ground states instead of the ground states themselves. This connection also presents possible clues forward for defining $\cE$ in non-Abelian theories.

\subsubsection*{Probing essential topological entanglement in general states}

We can also discuss generic pure and mixed states. The calculation of the essential topological entanglement in these cases follows easily, however the coefficients of superposition will typically pollute the topological aspects of the entanglement.\footnote{The exception being if those coefficients are topological themselves as in the states discussed in the previous subsection.} This also occurs in similar calculations of the TEE using generic pure states; see e.g. \cite{Wen:2016snr} for example calculations. One benefit of the algebraic approach to entanglement is that it is clear how the structure of the subregion algebras and their central elements lead to $\cE$; it would be useful if this structure could be utilized to isolate $\cE$ cleanly in arbitrary states. One such structure is that $\fA_\t{mag}$ and $\fA_\t{elec}$ appear as \emph{complementary operator algebras} \cite{Magan:2020ake}. Viewing $\fA_\t{greedy}=\fA_\t{elec}\vee\fA_\t{mag}$, there is a related structure of complementary conditional expectations, $E$ and $E'$,
\begin{equation}
	\begin{tikzcd}
		\fA_\t{greedy}[\aent] \arrow{r}{E} \arrow{d}{\mathsf{c}} & \fA_\t{mag}[\aent] \arrow{d}{\mathsf{c}} \\%
		\fA_\t{aus}[\co\aent] \arrow{u} &\arrow[swap]{l}{E'} \arrow{u} \fA_\t{mag}[\co\aent]
	\end{tikzcd}
\end{equation}
(and similarly for $\fA_\t{elec}$). One can then try to use entropic certainty relations to place strict bounds on the relative entropies in terms of the index of $[\fA_\t{mag}\colon\fA_\t{greedy}]$ \cite{Magan:2020ake}. As of yet we have been unable to utilize this technology to constrain the entanglement entropy of generic pure states: it is likely possible to construct pure states whose algebraic entropy saturates $\log\dim\mc H_\Sigma$ and so washes out the more intricate features of $\cE$.  Regardless, this avenue and the related avenue of the topological uncertainty principle \cite{Jian:2015wra} are worth exploring further.

\subsubsection*{Finer measures of essential topological entanglement}

In this paper we have focussed on the entanglement entropy of a reduced ground state expressed in a basis of surface operators. This reduced state contains potentially much more information than its von Neumann entropy. It is worth speculating on what can be learned from the finer structures of reduced states in the context of essential topological entanglement. One intriguing avenue is for an entanglement-based classification of $(2+1)$ dimensional topological order beyond modular data: it was demonstrated in \cite{mignard2021modular} that there arbitrarily many examples of distinct modular tensor categories with equivalent modular data. Many of these examples, however, can be distinguished by representations of the mapping class groups of genus-2 surfaces \cite{Wen:2019ylt}. Due to its sensitivity to global topological structures and global braiding of surface operators, this presents a tantalizing advantage of essential topological entanglement over the traditional TEE (which is only sensitive to the number of connected components of the entangling surface in $(2+1)$ dimensions). Our discussion of multi-boundary entanglement above already provides some evidence towards this statement: it has been demonstrated that all of the counterexamples of \cite{mignard2021modular} can be distinguished through link invariants \cite{Bonderson:2018ryx,kulkarni2018topological}. As it stands now, the entanglement spectrum of the reduced ground states in this paper is too simple to be utilized in this fashion. Perhaps multi-partite entanglement measures (such as the entanglement negativity), more appropriately chosen states (such as those in the multi-boundary set-ups), or the essential entanglement spectrum generalized to non-Abelian ground states can provide more leverage. More broadly, the role of correlations that the algebra of higher-form operators induce on states in distinguishing topological order in $(2+1)$ or even in higher dimensions is a future research direction worth pursuing.\footnote{We thank an anonymous referee for suggesting this direction to us.}

\subsubsection*{Applications beyond BF theory: fractons}

As mentioned above, much of this paper and its follow-up is motivated by the question of topological order in higher dimensions. While our focus has been on standard Abelian topological orders, we hope some of our ideas translate to $(3+1)$ gapped fracton phases. This translation is most easily facilitated through the \textquote{foliated field theory} framework to describe Type I, or foliated, fracton order \cite{Slagle:2020ugk}. Fracton phases, foliated phases included, have interesting forms of UV/IR mixing that make the distinction between different UV scales (e.g. the energy cutoff, the momentum cutoff, and the lattice scale) subtle and important. Essential topological entanglement eschews at least some of this subtlety: it does not rely on a UV embedding, but instead utilizes only the structure of symmetry operators (which may still rely on a lattice scale for foliated fracton phases). It would be very interesting if essential topological entanglement can provide a more natural way to extract universal features of foliated fracton phases directly in the continuum.

\subsubsection*{Essential topological entanglement in generic theories}

Although we have focused on essential topological entanglement in Abelian BF theory, the extension to other TQFTs is conceptually straightforward.  However, essential topological entanglement might prove to be a useful concept in generic (non-topological) quantum field theories exhibiting generalized global symmetries.  This follows from the realization that symmetries, and more broadly generalized symmetries, are synonymous with topological operators (of various codimensions, invertible or non-invertible) \cite{Gaiotto:2014kfa,Bhardwaj:2017xup}.\footnote{There is a multitude of results stemming from this realization, see e.g. \cite{Cordova:2022ruw} for a more complete list.} One can then define topological operator algebras in a generic quantum field theory, by restricting to the algebras of symmetry operators. More formally, the sandwich approach \cite{Freed:2022qnc,Gaiotto:2020iye} to global symmetries, provides an avenue to delineate these operators from the rest of the theory: all symmetry operators live in the one-higher-dimensional SymTFT \cite{Apruzzi:2021nmk}, which is a topological field theory in its own right. Therefore the essential topological entanglement applied to this SymTFT is a potential probe of an indistillable, symmetry-induced, entanglement of the original theory.

Along these lines, we can lastly speculate on consequences for gravity. It is strongly believed that quantum gravity has no global symmetries (see e.g. \cite{Kallosh:1995hi,Banks:2010zn,Harlow:2018tng}), although there are exceptions in low-dimensional models excluding black holes \cite{Harlow:2020bee}. It is tempting to phrase this condition in the language of entanglement and conjecture that \emph{quantum gravity has no essential topological entanglement,} which might be a weaker, but more universal condition on quantum gravity.

\section*{Acknowledgements}

We thank Sean Hartnoll, Diego Hofman, Diego Liska, Onkar Parrikar, and Ronak Soni for fun and enlightening discussions.  JRF thanks Rob Leigh and Matthew Lapa for conversations inspiring this work. JRF thanks the University of Amsterdam for hospitality. SV thanks the University of Cambridge 
and the Kavli Instintute for Theoretical Physics at UCSB for hospitality. Research at KITP was supported in part by the National Science Foundation under Grant No. NSF PHY-1748958. JRF is supported by STFC consolidated grant ST/T000694/1 and by Simons Foundation Award number 620869. SV is supported by the NWO Spinoza prize awarded to Erik Verlinde.

\appendix
\section{Precision BF}\label{app:precBF}
\newcommand{\hol}{\operatorname{hol}}
\newcommand{\bil}{\mathrm{q}}

In this appendix we take a closer look at BF theory in cases where the usual definitions are insufficient, such as when the theory is defined on a manifold with torsion. The appropriate language to define BF theory precisely is that of differential cohomology. We use this language to write down the BF action on a generic manifold and give a precise definition of the path-integral measure and its properties. For the following we will take \(X\) to be a \(d\)-dimensional manifold, possibly with boundary and we will denote by \(\ib:\pd X\hookrightarrow X\) the embedding of the boundary.

When studying \(p\)-form gauge theories on a non-trivial manifold the \(p\)-form gauge field is insufficient to capture all the topological properties of the theory. Instead, the relevant degrees of freedom can be captured by a $(p+1)$-cocycle in differential cohomology, $\dH^{p+1}(X)$,\footnote{More precisely, differential cohomology is categorical in spirit, so this is a workable model of it, known as Deligne or Deligne--Beilinson cohomology.} i.e. a triplet $\check{A}=(A,N_A,F_A)$, where $A$ is a regular $p$-cochain (the gauge field), $N_A$ is a $(p+1)$-cocycle (giving rise to the flux, upon integration) and $F_A$ is a closed $(p+1)$-form (the field strength). The constituents of the triplet are related by a constraint: $F_A = \dd{A}+N_A$. To define actions, we also need a product in differential cohomology, \(\ve:\dH^p(X)\times\dH^q(X)\to\dH^{p+q}(X)\), that will replace the usual wedge product of differential forms. For a gentle introduction on the usage of differential cohomology in higher-gauge theories see \cite{Hsieh:2020jpj}, while for a mathematically rigorous approach see  \cite{amabel2021differential}. The bottomline, is that using the product \(\vee\), one can write the BF action, on a generic \(d\)-dimensional manifold $X$, as a pairing in differential cohomology:
\begin{equation}
	S_\BF\qty[\check{A},\check{B}]\coloneqq\int_{X} \check{B}\ve \check{A}, \label{eq:pair}
\end{equation}
The above action action can we written in more user-friendly way as
\begin{equation}
	S_\BF\qty[\check{A},\check{B}] = \int_X B_\t{flat}\w \dd{A} + \int_{\mathrm{PD}[N_A]} B, \label{eq:pair-user-friendly}
\end{equation}
where \(B_\t{flat}\) is, as the name suggests the flat part of \(B\), with \(F_{B_\t{flat}}=0\) and \(\mathrm{PD}[N_A]\in\H_{d-p-1}(X)\) is the Poincaré dual of \(N_A\). Here let us note that while differential cohomology is the correct framework to define BF theory, in practice only the torsion part of \(X\) can have a non-trivial effect. This is most readily seen in \Cref{eq:pair-user-friendly}. Since \(\H_{d-p-1}(X)\cong \mathrm{Tor}_{d-p-1}(X)\oplus\Z^{\b_{d-p-1}(X)}\), Dirac quantization implies that the non-torsion part of the second term contributes as \(2\pi\times\t{integer}\), which is trivial upon exponentiation. Therefore the action is effectively \(S=\int_X B\w\dd{A} +\t{torsion}\), reducing to the usual BF action, \Cref{eq:BFact}, whenever \(X\) has no torsion.

The final step to close this intellectual detour is to define the BF path-integral. The formal path-integral measure is
\begin{equation}
	\dd{\mu}(\check{A},\check{B}) \coloneqq \DD{\check{A}}\DD{\check{B}} \ex{\ii S_\BF\qty[\check{A},\check{B}]}, \label{eq:measure}
\end{equation}
where we are summing over differential cohomology elements. The functional measure \(\DD{\check{A}}\) instructs us to integrate over all closed \((p+1)\)-forms \(F_A\in\Ocl^{p+1}(X)\), integrate over all \(p\)-cochains \(A\in\CC^p(X)\), and finally sum over all \((p+1)\)-cohomology classes, \(N_A\in\H^{p+1}(X)\), respecting the constraint \(F_A = \dd{A}+N_A\). Moreover, to avoid overcounting we need to identify configurations that differ by a flat \((p-1)\)-form field that vanishes on the boundary. The latter we will do simply by dividing out by the volume of these gauge transformations. The integral over \(F_A\) is trivial, due to the constraint, so we are left with
\begin{equation}
	\int_{\dH^{p+1}(X)} \DD{\check{A}} \cO[\check{A}] = \sum_{N_A\in\H^{p+1}(X;\Z)} \int_{\CC^p(X)} \frac{\DD{A}}{\vol(\cG_p(X))} \cO\qty[\qty(A,N_A,\var{A}+N_A)],
\end{equation}
where \(\cO[\check{A}]\) is an arbitrary test-functional.

We will only be dealing with a continuous structure group, therefore, we can safely regard \(A\) as a \(p\)-form. In this case, \(\cG_p\) is defined recursively as (taking \(\cG_0 = \emptyset\))
\begin{equation}
	\cG_p(X) \coloneqq \set{\alpha\in \Omega^{p-1}(X)\suchthat \ib^* \alpha = 0}\Big/\cG_{p-1}(X).
\end{equation}
Note that the gauge transformations act in principle also on $N_A$. Their action has been, however, absorbed into constraining \(N_A\) to be in the cohomology rather than in \(\ZZ^{p+1}(X;\Z)\). Obviously, the story is the same for \(\DD{\check{B}}\), with the appropriate changes in form-degrees.

One important property of the measure \Cref{eq:measure}, which we use in the main text, is that if we shift \(\check{A}\) and \(\check{B}\) by some fixed \(\check{\alpha}\) and \(\check{\beta}\) respectively, the measure transforms as
\begin{equation}
	\dd{\mu}(\check{A}+\check{\alpha},\check{B}+\check{\beta}) = \dd{\mu}(\check{A},\check{B})\ \ex{\ii S_\BF\qty[\check{\alpha},\check{B}]}\ \ex{\ii S_\BF\qty[\check{A},\check{\beta}]}\ \ex{\ii S_\BF\qty[\check{\alpha},\check{\beta}]}. \label{eq:property}
\end{equation}

\section{Algebras with trivial center}\label{app:centerless}

As has been mentioned at several points in the main text, one hallmark of gauge theories is the presence of centers in algebras assigned to subregions.  While this feature is generic, the assignment of an algebra, $\fA[\aent]$, to a region, $\aent$, is ultimately a choice.  As such one may ask if there is enough ambiguity to assign a centerless algebra to $\aent$. To be clear, of course one can always assign, by fiat, a centerless algebra to $\aent$; for two (extreme) instances:
\begin{equation}
	\fA[\aent]=\t{span}_\bbC\set{1}\qquad\t{or}\qquad\fA[\aent]=\fA[\Sigma],
\end{equation}
where $\fA[\Sigma]$ is the full algebra of operators\footnote{We assume that $\fA[\Sigma]$ has no center.  This is true for the algebra of surface operators in this section.  However, in lattice gauge theories, it is not typically true: there are an extensive number of Gauss operators that generate a global center \cite{Lin:2018bud}. Gauge-invariant states in the BF theory (treated as an IR EFT of such a lattice gauge theory) live in a fixed eigenvalue sector of this \textquote{microscopic center.}} on $\mc H_\Sigma$. However, these assignments tell us absolutely no physics about the region in which we are interested!  Barring such \textquote{madman assignments,} centerless subalgebras associated to a subregion $\aent$ are highly non-generic in gauge theories, yet they may still exist. We can attempt to build such an subalgebra by starting from a generic (centerful) algebra, $\fA[\aent]$, and either (i) systematically excluding operators from $\fZ[\aent]$ \cite{Casini:2013rba} (a process we will call \emph{reduction}), or (ii) systematically adding to $\fA[\aent]$ operators from $\fA[\Sigma]$ that do not commute with $\fZ[\aent]$ (what we will call \emph{extension}).  In the course of both of these processes, the resulting $\fA[\aent]$ will have a more tenuous relationship to its associated region $\aent$. However, to maintain at least some degree of association between $\fA[\aent]$ and $\aent$ we will focus on reductions and extensions that are minimal to ensuring $\fA[\aent]$ is centerless.

Starting then from $\fA_\t{mag}[\aent]$, minimal reduction and minimal extension result in two centerless algebras associated to $\aent$, respectively:

\subsubsection*{The austere algebra}

For the process of reduction, we remove from $\fA_\t{mag}[\aent]$ the surface operators homotopic to $\pa\aent$.  Since these are also the operators deformable to being contained in $\coaent$, this equivalent to
\begin{equation}
	\fA_\t{aus}[\aent]\coloneqq\fU\set{\hat{\W}_{\eta^\sfi}^{\vec w_\sfi},\hV_{\sigma^\sfj}^{\vec v^\sfj} \suchthat \eta^\sfi,\,\sigma^\sfj\hin \aent\;\t{and}\;\eta^\sfi,\sigma^\sfj\nothin \coaent}=\fA_\t{mag}[\aent]\cap\fA_\t{elec}[\aent],
\end{equation}
i.e. $\fA_\t{aus}[\aent]$ is generated only by surface operators that, homotopically, \emph{must} be completely in $\aent$.

\subsubsection*{The greedy algebra}
For extension, we now add to $\fA_\t{mag}[\aent]$ the minimal basis of surface operators that \textquote{pierce} $\pa\aent$.  These link with operators on $\pa\aent$ and so prohibit them from forming a center:
\begin{equation}  \label{eq:Agreedy}
	\fA_\t{greedy}[\aent]\coloneqq\fU\set{\hat{\W}_{\eta^\sfi}^{\vec w_\sfi},\hV_{\sigma^\sfj}^{\vec v_\sfj} \suchthat \eta^\sfi,\,\sigma^\sfj\hin \aent\;\t{or}\;\eta^\sfi,\sigma^\sfj\nothin \coaent}=\fA_\t{mag}[\aent]\cup\fA_\t{elec}[\aent]
\end{equation}
i.e. $\fA_\t{greedy}[\aent]$ is generated by all surface operators that are either homotopically in $\aent$, or topologically must have \textquote{a leg in $\aent$}.\footnote{to be precise \(\fA_\t{greedy}[\aent]\) should be the smallest algebra containing the two. That is the generated algebra, $\fA_\t{mag}[\aent]\vee\fA_\t{elec}[\aent]\coloneqq\cco{\qty(\fA_\t{mag}[\aent]\cup\fA_\t{elec}[\aent])}$. But for finite-dimensional operator algebras, as is our case, it coincides with \Cref{eq:Agreedy}, since \(\cco{\fA}=\fA\) is automatically guaranteed. See also \cite{Casini:2019kex}.}

It is clear the same algebras can be constructed from the reduction and extension, respectively, of $\fA_\t{elec}[\aent]$, as well.  It is also clear that
\begin{equation}
	\co{\left(\fA_\t{aus}[\aent]\right)}=\fA_\t{greedy}[\coaent],\qquad\qquad \co{\left(\fA_\t{greedy}[\aent]\right)}=\fA_\t{aus}[\coaent].
\end{equation}

Because these two subalgebras have trivial center, they should correspond to honest tensor factorizations of $\mc H_\Sigma$.  This is indeed the case.  Returning to the partition of a general state, \Cref{eq:MCstatepartition1}, for the austere algebra, we simply group the charges associated with central surface operators with $\coaent$
\begin{equation}
	\ket{\fv}\coloneqq\ket{\set{\fv^{\inr{\aent}}},\set{\fv^{{\coaent}^+}}}\qquad\qquad \set{\fv^{{\coaent}^+}}\coloneqq\set{\fv^{\pa\aent}_{\perp},\fv^{\pa\aent}_{\parallel},\fv^{\inr{\coaent}}},
\end{equation}
while for the greedy algebra, we group them in with $\aent$:
\begin{equation}
	\ket{\fv}\coloneqq\ket{\set{\fv^{\aent^+}},\set{\fv^{\inr{\coaent}}}}\qquad\qquad \set{\fv^{{\aent}^+}}\coloneqq\set{\fv^{\inr{\aent}},\fv^{\pa\aent}_{\perp},\fv^{\pa\aent}_{\parallel}}.
\end{equation}
Correspondingly the Hilbert space decomposes as
\begin{equation}\label{eq:centerlessHSfactorization}
	\mc H_\Sigma=\mc H_{\inr{\aent}}\otimes \mc H_{{\coaent}^+}\qquad\t{or}\qquad\mc H_\Sigma=\mc H_{\aent^+}\otimes \mc H_{\inr{\coaent}},
\end{equation}
under the action of $\fA_\t{aus}[\aent]$ or $\fA_\t{greedy}[\aent]$, respectively.  There is a cost for being centerless: our centerless algebras, $\fA_\t{aus}[\aent]$ and $\fA_\t{greedy}[\aent]$, have looser relationships to their associated region.  Correspondingly, Hilbert space decompositions such as \Cref{eq:centerlessHSfactorization} may possess less information about a region than decompositions possessing a center.  Additionally, even though we have arrived at $\fA_\t{aus(greedy)}[\aent]$ through minimal reduction (extension), owing to the topological nature of the theory and its operator content, the dissociation from $\aent$ may still be drastic indeed.  For example, as illustrated in  \Cref{fig:3t-eon}, there are situations (namely when either $\fA[\aent]=\fZ$ or $\fA[{\coaent}]=\fZ$) where minimal extension or subtraction can result in everything-or-nothing centerless algebras
\begin{equation}\label{eq:extremedissociation}
	\fA_\t{aus}[\aent]=\t{span}_\bbC\set{1}\qquad\t{or}\qquad \fA_\t{greedy}[\aent]=\fA[\Sigma],
\end{equation}
even when $\aent$, $\fA_\t{mag}[\aent]$, and $\fA_\t{elec}[\aent]$ are non-trivial.  In such situations, insofar as the entanglement entropy is concerned, the cost of being centerless is then very heavy: it is zero for \emph{all} pure states.

\begin{figure}[!tb]
	\centering
    \def\svgwidth{0.6\textwidth}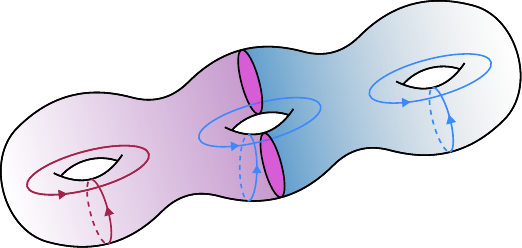
	\caption{The austere algebra $\fA_\t{aus}[\aent]$ and the greedy algebra $\fA_\t{greedy}[\coaent]$, on a 3-torus, generated by a basis of operators along longitude and meridian cycles, $\set{\ell_i,m_i}_{i=1,2,3}$. The color-coding follows that of \cref{fig:3t}.}
	\label{fig:3t-app}
\end{figure}
\begin{figure}[!tb]
	\centering
	\def\svgwidth{0.6\textwidth}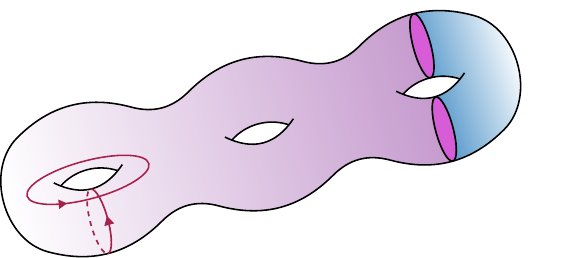
	\caption{An example of \textquote{everything-or-nothing} centerless algebras: $\fA_\t{greedy}[\aent]=\fA[\Sigma]$, $\fA_\t{aus}[\coaent]=\operatorname{span}_\C\set{1}$.}
	\label{fig:3t-eon}
\end{figure}

\section{Proofs for \texorpdfstring{\Cref{sect:OAent}}{Section \ref{sect:OAent}}.}\label{app:zm}

\subsection*{Proof of subregion electric-magnetic duality}

Subregion electric-magnetic duality as described in \Cref{sect:topsubalgs} is the statement that the surface operators generating $\fA_\t{mag}[\aent]$ are in 1-1 correspondence with surface operators generating $\fA_\t{elec}[\aent]$.
Let us prove it.  To do so, it suffices to show the following:
\begin{claim}
	The pairing
	\begin{align*}
		\bbL:\H_p(\Sigma)\times \H_{d-p-1}(\Sigma) & \longrightarrow \R               \\
		(\eta,\sigma)                              & \mapsto \int_{\eta\cap \sigma} 1
	\end{align*}
	is non-degenerate restricted to \(\im i^\aent_p\times \coker i^{\coaent}_{d-p-1}\), where, as explained in the main text \(i^\aent_k:\H_k(\aent)\to\H_k(\Sigma)\) is the pushforward in homology of the map that embeds \(\aent\) into \(\Sigma\) and similarly for \(\coaent\).
\end{claim}
\begin{proof}
	Let us assume the opposite of the claim. That is, either,
	\begin{enumerate}[label=(\roman*)]
		\item \label{case:1-EM} There exists an \(\hat{\eta}\in \im i^\aent_p\), such that \(\bbL(\hat{\eta},\sigma)=0\), for all \(\sigma\in\coker i^R_{d-p-1}\) or
		\item \label{case:2-EM} There exists a \(\hat{\sigma}\in \coker i^\aent_p\), such that \(\bbL(\eta,\hat{\sigma})=0\), for all \(\eta\in\im i^R_{d-p-1}\).
	\end{enumerate}
	Note that since \(\eta\in\im i^\aent_p\), it can always be homotopically deformed to lie entirely in \(\aent\). We can therefore restrict our attention to cycles restricted in \(\aent\). Then \(\left.\eta\right|_\aent\in\H_p(\aent)\), while \(\left.\sigma\right|_\aent \in \H_{d-p-1}(\aent,\pd\aent)\).

	Let us assume case \ref{case:1-EM}. By Poincaré--Lefschetz duality, \(\H_{d-p-1}(\aent,\pd\aent) \cong \H^p(\aent\setminus\pd\aent)\)\linebreak\(\cong \H^p\qty(\inr{\aent})\), we have
	\begin{equation}
		0 = \int_{\left.\hat{\eta}\right|_\aent \cap \left.\sigma\right|_\aent} 1 = \int_{\left.\hat{\eta}\right|_\aent} \pdual{\left.\sigma\right|_\aent}, \qquad\fall\,\pdual{\left.\sigma\right|_\aent}\in\H^p\qty(\inr{\aent}),
	\end{equation}
	where \(\pdual{\cdot}\) denotes the Poincaré(--Lefschetz) dual of a cycle. In words, there exists a \(p\)-cycle of \(\aent\) that is orthogonal to all \(p\)-cohomology classes of \(\inr{\aent}\), which is impossible. The proof of \ref{case:2-EM} is wholly similar, with the difference being that the conclusion is that there exists a \(p\)-cohomology class in \(\inr{\aent}\) orthogonal to all \(p\)-cycles of \(\aent\), which is again impossible.
\end{proof}

\subsection*{Proof of claim \Cref{eq:claimmag}}

\begin{claim}
	\begin{equation}
		\co{\qty(\fA_\t{mag}[\aent])} = \fU\set{\hat{\W}_{\eta^\sfi}^{\vec w_\sfi},\hV_{\sigma^\sfj}^{\vec v_\sfj} \suchthat  \eta^{\sfi},\sigma^{\sfj}\hin \coaent}\equiv\fA_\t{mag}\qty[\co\aent].
	\end{equation}
\end{claim}
\begin{proof}
	Consider first the commutant of the operators \(\hV_\sigma^{\vec{v}}\in\fA_\t{mag}[\aent]\). We want to show that \(\hW_\eta^{\vec{w}}\) commutes with all \(\hV_\sigma^{\vec{v}}\in\fA_\t{mag}[\aent]\) iff \(\eta\in\im i^{\coaent}_p\). Firstly note that by the definition of the algebra, \Cref{eq:WValg}, \(\hW_\eta^{\vec{w}}\) commutes with \(\hV_\sigma^{\vec{v}}\) iff \(\bbL(\eta,\sigma)=0\). The statement we need to prove reduces, then, to the following:
	\begin{equation}
		\bbL(\eta,\sigma)=0,\ \fall\sigma\in\im i^\aent_{d-p-1} \qq{iff} \eta\in\im i^{\coaent}_p.
	\end{equation}
	\begin{itemize}
		\item[\((\Rightarrow)\)] If \(\eta\in\im i_p^{\coaent}\), it is clear that \(\bbL(\eta,\sigma)=0,\ \fall\sigma\in\im i^\aent_{d-p-1}\), since we can homotopically move \(\sigma\) and \(\eta\) to lie within the interior of \(\aent\) and \(\coaent\) respectively.
		\item[\((\Leftarrow)\)] We will prove the only if direction ad absurdum. For that, suppose that \(\eta\notin\im i^{\coaent}_p\). Then, since \(\H_p(\Sigma)=\im i^{\coaent}_p\oplus \coker i^{\coaent}_p\), \(\eta\in\coker i^{\coaent}_p\). Restricted to \(\aent\) then, \(\left.\eta\right|_{\aent}\in\H_p(\aent,\pd\aent)\) and \(\left.\sigma\right|_\aent\in\H_{d-p-1}(\aent)\). By Poincaré--Lefschetz duality, $\H_p(\aent,\pd\aent)\cong \H^{d-p-1}(\aent\setminus \pd\aent)=\H^{d-p-1}\qty(\inr{\aent})$ and we have then,
			\begin{equation}
				0 \demeqq \bbL(\eta,\sigma) = \bbL\qty(\left.\eta\right|_\aent,\left.\sigma\right|_\aent) = \int_{\left.\sigma\right|_\aent} \pdual{\left.\eta\right|_\aent},\ \fall\sigma\in\im i^p_\aent.
			\end{equation}
			where, \(\pdual{\left.\eta\right|_\aent}\in \H^{d-p-1}\qty(\inr{\aent})\) is the Poincaré--Lefschetz dual of \(\left.\eta\right|_\aent\). That is to say, \(\pdual{\left.\eta\right|_\aent}\) must be orthogonal to all \((d-p-1)\)-cycles of \(\aent\), which is impossible. Therefore the assumption that \(\eta\notin\im i^{\coaent}_p\) was absurd.
	\end{itemize}
	The proof for the commutant of operators $\hW_{\eta}^{\vec w}\in\fA_\t{mag}[\aent]$ is wholly similar, concluding thus the proof of the claim.
\end{proof}

\subsection*{Proof of claim \Cref{eq:claimelec}}

\begin{claim}
	\begin{equation}
		\co{\left(\fA_\t{elec}[\aent]\right)}=\fU\set{\hat{\W}_{\eta^\sfi}^{\vec w_\sfi},\hV_{\sigma^\sfj}^{\vec v^\sfj}\suchthat \eta^\sfi,\sigma^\sfj\nothin\aent}=\fA_\t{elec}[\coaent].
	\end{equation}
\end{claim}
\begin{proof}
	Consider first the commutant of the operators \(\hV_\sigma^{\vec{v}}\in\fA_\t{elec}[\aent]\). We want to show that \(\hW_\eta^{\vec{w}}\) commutes with all \(\hV_\sigma^{\vec{v}}\in\fA_\t{elec}[\aent]\) iff \(\eta\in\coker i^{\aent}_p\). As before, by the definition of the algebra, \Cref{eq:WValg}, the statement we need to prove reduces to:
	\begin{equation}
		\bbL(\eta,\sigma)=0,\ \fall\sigma\in\coker i^{\coaent}_{d-p-1} \qq{iff} \eta\in\coker i^{\aent}_p.
	\end{equation}
	\begin{itemize}
		\item[\((\Leftarrow)\)] We will first prove the only if direction. This we will prove again by contradiction. For that, suppose that \(\eta\notin\coker i^{\aent}_p\). Then, \(\eta\in\im i^{\aent}_p\). Restricting to \(\aent\) we have \(\left.\eta\right|_{\aent}\in\H_p(\aent)\) and \(\left.\sigma\right|_\aent\in\H_{d-p-1}(\aent,\pd\aent)\). By Poincaré--Lefschetz duality, $\H_{d-p-1}(\aent,\pd\aent)\cong \H^{p}(\aent\setminus \pd\aent)=\H^{p}\qty(\inr{\aent})$ hence,
			\begin{equation}
				0 \demeqq \bbL(\eta,\sigma) = \bbL\qty(\left.\eta\right|_\aent,\left.\sigma\right|_\aent) = \int_{\left.\eta\right|_\aent} \pdual{\left.\sigma\right|_\aent},\ \fall\sigma\in\coker i^p_\aent.
			\end{equation}
			where, \(\pdual{\left.\sigma\right|_\aent}\in \H^{p}\qty(\inr{\aent})\) is the Poincaré--Lefschetz dual of \(\left.\sigma\right|_\aent\). That is to say, every \(p\)-cocycle in the interior of \(\aent\) must be orthogonal to \(\left.\eta\right|_\aent\), which is impossible. Therefore the assumption that \(\eta\notin\coker i^{\coaent}_p\) was absurd.
		\item[\((\Rightarrow)\)] The if direction follows immediately from subregion electric-magnetic duality, proven above, and the fact that the rank of \(\bbL\) is \(\b_p(\Sigma)\) (since it is inverse to \(\bbG_p\) as a matrix). The restrictions of \(\bbL\) to the subspaces \linebreak\(\im i^\aent_p\times\coker i^{\coaent}_{d-p-1}\) and \(\coker i^\aent_p\times \im i^{\coaent}_{d-p-1}\) saturate the rank of \(\bbL\), hence the rank of \(\bbL\) restricted to \(\coker i^\aent_p\times\coker i^{\coaent}_{d-p-1}\) is zero. In other words for any \(\eta\in\coker i^\aent_p\), \(\bbL(\eta,\sigma)=0\), \(\fall\sigma\in\coker i^{\coaent}_{d-p-1}\), concluding the proof.
	\end{itemize}
	The proof for the commutant of operators $\hW_{\eta}^{\vec w}\in\fA_\t{elec}[\aent]$ is wholly similar, concluding thus the proof of the claim.
\end{proof}

\subsection*{Counting the magnetic and electric centers}

Let us first focus on the magnetic center.
Consider, as in the main text, the pushout square that embeds \(\pd\aent\) into \(\Sigma\) through \(\aent\) and \(\coaent\):
\[\begin{tikzcd}[ampersand replacement=\&]
		\pd\aent \& \aent \\
		\coaent \& \Sigma
		\arrow["j^{\aent}", from=1-1, to=1-2]
		\arrow["i^\aent", from=1-2, to=2-2]
		\arrow["i^{\coaent}"', from=2-1, to=2-2]
		\arrow["j^{\coaent}"', from=1-1, to=2-1]
		\arrow["\ell^{\pd\aent}", from=1-1, to=2-2]
	\end{tikzcd}\]
Each of these maps induces a push-forward on the homology,
\begin{align}
	j^{{\aent}^{(\sfc)}}_k & : \H_k\qty(\pd{\aent}^{(\sfc)}) \to \H_k\qty({\aent}^{(\sfc)}) \\
	i^{{\aent}^{(\sfc)}}_k & : \H_k\qty(\aent^{(\sfc)}) \to \H_k\qty(\Sigma).
\end{align}
The operators in the magnetic algebra of \(\aent\), \(\fA_\t{mag}[\aent]\), are generated by surface operators whose cycles lie in the image of \(i^{\aent}_{p}\) and \(i^{\aent}_{d-p-1}\) and similarly for \(\fA_\t{mag}[\coaent]\), replacing \(\aent\) by \(\coaent\). The cycles generating the center then lie in \(\im i^{\aent}_{p} \cap \im i^{\coaent}_{p}\) and \(\im i^{\aent}_{d-p-1} \cap \im i^{\coaent}_{d-p-1}\). Since $\aent$ and $\coaent$ share only $\pa\aent$, it is evident that the cycles generating the center lie in \(\im\qty(i^\aent\circ j^{\aent})_k\cong \im\qty(i^{\coaent}\circ j^{\coaent})_k\eqqcolon \im\ell^{\pd\aent}_k\), as illustrated by the above pushout square. Therefore, the dimension of the magnetic center is
\begin{equation}
	\abs{\fZ_\t{mag}[R]} = \abs{\det\bbK}^{\left(\h^p_\t{mag}+\h^{d-p-1}_\t{mag}\right)},\qquad\qquad \h^k_\t{mag}=\dim\im\ell_k^{\pa\aent}.
\end{equation}
Since \(\pd\aent\subset \Sigma\), we can calculate \(\dim\im\ell^{\pd\aent}_k\) by the long exact sequence of relative homology \cite{hatcher}:
\begin{equation}\label{eq:LES-dR-Sigma}
	\cdots\to\H_k(\pd \aent)\xto{\ell_k}\H_k(\Sigma)\xto{r_k}\H_k(\Sigma,\pd\aent)\xto{\delta_{k-1}}\H_{k-1}(\pd \aent)\to\cdots.
\end{equation}
Using exactness of the sequence, we find that
\begin{align}\label{eq:dimellk}
	\dim\im\ell_k & = (-1)^{k-1}\sum_{n=0}^{k-1} (-1)^n \dim\H_n(\pd\aent) \nn & \phantom{=~} +(-1)^k\sum_{n=0}^{k}(-1)^n\qty\big(\dim\H_n(\Sigma)-\dim\H_n(\Sigma,\pd\aent)),
\end{align}
which leads to \Cref{eq:hpmag}.

It is useful to pause at this point and illustrate how the bulk dependent terms of \Cref{eq:dimellk} ensure the correct counting of operators when $\Sigma$ is topologically trivial.  For example, consider the case when $\Sigma=\S^2$, $\aent$ is collection of $q$ disks such that $\pa\aent=\bigsqcup^q\S^1$, and let $k=1$.  In this case, since there are no non-trivial 1-cycles on $\S^2$, there are no 1-cycle operators in $\fA[\Sigma]$ to count.  We should find $\dim\im\ell_1=0$.  Applying \Cref{eq:hpmag} we find
\begin{align}
	\dim\im\ell_1 & =\dim\H_0(\pa\aent)-\dim\H_0\qty(\S^2)+\dim\H_0\qty(\S^2,\pa\aent)+\dim\H_1\qty(\S^2)-\dim\H_1\qty(\S^2,\pa\aent)\nonumber \\
	              & =q-1+\dim\H_0\qty(\S^2,\pa\aent)-\dim\H_1\qty(\S^2,\aent)
\end{align}
To calculate the dimensions of the relative homologies, we note that $\H_n(\Sigma,\pd\aent)\cong\widetilde{\H}_n(\Sigma/\pd\aent)$ where $\widetilde{\H}_n(\cdot)$ denotes reduced homology and $\Sigma/\pa\aent$ is the quotient space.  By definition of the reduced homology $\dim\widetilde{\H}_0\qty(\S^2/\pa\aent)=0$ and we can easily calculate $\dim\widetilde{\H}_1\qty(\S^2/\pa\aent)=q-1$, as illustrated in \Cref{fig:bouquet}.  Thus indeed
\begin{equation}
	\dim\im\ell_1=0
\end{equation}
in this case.

\begin{figure}[!bht]
	\centering
	\def\svgwidth{0.7\textwidth}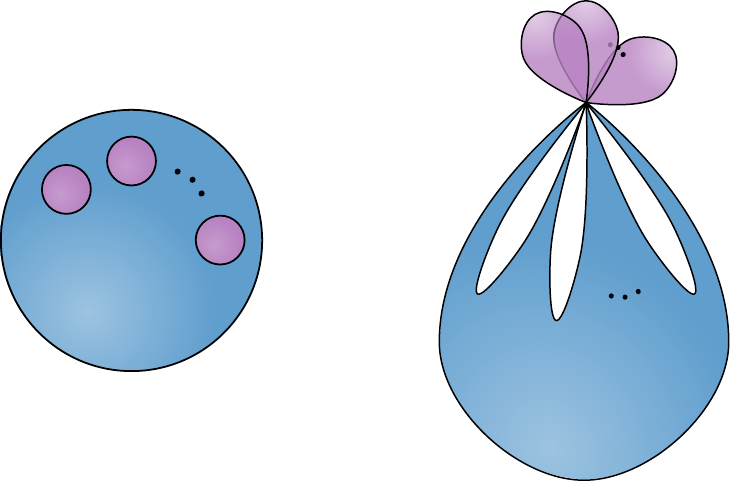
	\caption{(Left) The setup of the above example: $\Sigma=\S^2$ and $\aent = \bigsqcup^q \bbB^2$. (Right) \(\Sigma/\pd\aent\).}
	\label{fig:bouquet}
\end{figure}

Coming back to the general case, we can now sum the \(\h_\t{mag}^p\) and  \(\h_\t{mag}^{d-p-1}\), to find
\begin{align}
	\h_\t{mag} & \coloneqq \dim\im\ell^{\pd\aent}_p + \dim\im\ell^{\pd\aent}_{d-p-1} \nn
	           & = 2\sum_{n=0}^{p-1}(-1)^{p-1-n}\dim \H_n(\pa\aent)+(-1)^{p-1}\upchi(\pa\aent)+\dim\H_p(\Sigma)+(-1)^p\upchi(\Sigma)\nn
	           & \qquad +\sum_{n=0}^p(-1)^{p-1-n}\dim\H_n(\Sigma,\pa\aent)+\sum_{n=0}^{d-p-1}(-1)^{d-p-n}\dim\H_n(\Sigma,\pa\aent)\nn
	           & = 2\sum_{n=0}^{p-1}(-1)^{p-1-n}\dim \H_n(\pa\aent)+(-1)^{p-1}\upchi(\pa\aent)+\dim\H_p(\Sigma)+(-1)^p\upchi(\Sigma)\nn
	           & \qquad +(-1)^{p-1}\upchi(\Sigma,\pa\aent)-\dim\H_p(\Sigma,\pa\aent)+(-1)^{d-p-1}\left(\dim\H_{d-1}(\Sigma,\pa\aent)-\dim\H_0(\Sigma,\pa\aent)\right)\nn
	           & = 2\sum_{n=0}^{p-1}(-1)^{p-1-n}\b_n(\pa\aent)+\left(\b_p(\Sigma)-\dim\H_p(\Sigma,\pa\aent)\right)\nn
	           & \qquad+(-1)^{d-p-1}\left(\dim\H_{d-1}(\Sigma,\pa\aent)-\dim\H_0(\Sigma,\pa\aent)\right).\label{eq:hmag-sum-app}
\end{align}
which leads to \Cref{eq:absZmag}.  In the second line above we've used $\H_n(\cdot)\cong\H_{D-n}(\cdot)$ for absolute homologies on $D$-dimensional compact spaces.  In the third line we've used the similar relation for the relative homology $\H_n(\Sigma,\pa\aent)$ which holds for all degrees except the top and bottom degrees, $\H_{d-1}(\Sigma,\pa\aent)$ and $\H_{0}(\Sigma,\pa\aent)$, respectively.\footnote{To see that, we note again that $\H_n(\Sigma,\pd\aent)\cong\widetilde{\H}_n(\Sigma/\pd\aent)$. For the reduced homology, it holds that $\widetilde{\H}_n(\cdot)\cong\H_n(\cdot)$, whenever \(n\neq 0\). On \(\Sigma/\pd\aent\), one can define a non-degenerate pairing $\bbL:\H_n(\Sigma/\pd\aent)\times\H_{D-n}(\Sigma/\pd\aent)\to\R$, as $\bbL(\alpha,\beta)\coloneqq\int_{\alpha\cap\beta} 1$. This renders $\H_{n}(\Sigma/\pd\aent)\cong\H_{D-n}(\Sigma/\pd\aent)$ whenever $n\notin\set{0,D}$ and shows the desired statement for relative homology.}  Additionally, \[\upchi(\Sigma,\pd \aent)\coloneqq\sum_{n=0}^{d-1}(-1)^n\dim\H_n(\Sigma,\pa\aent)\] is the relative Euler characteristic of the pair \((\Sigma,\pd\aent)\), and \(\upchi(\pd\aent)-\upchi(\Sigma)+\upchi(\Sigma,\pd\aent) =0\), as it is simply the rank-nullity relation of \Cref{eq:LES-dR-Sigma}. Lastly, we've used $\dim\H_n(\cdot)=\dim\H^n(\cdot)\equiv\b_n(\cdot)$ for absolute homologies on compact spaces.

For counting the dimension of $\fZ_\t{elec}[\aent]$, let us show
\begin{equation}
	\h_\t{elec}^p=\h_\t{mag}^{d-p-1}.
\end{equation}
That is, the number of $p$-cycles, $\eta$, such that $\eta\nothin\aent,\nothin\coaent$ is equal to the number of $(d-p-1)$-cycles, $\sigma$ with $\sigma\hin\aent,\hin\coaent$. To show this we need is suffices to show the following.

\begin{claim}
	The pairing
	\begin{equation}
		\bbL(\eta,\sigma)\coloneqq\int_{\eta\cap\sigma}1
	\end{equation}
	is non-degenerate restricted to $(\coker i^{\aent}_p\cap \coker i^{\coaent}_p)\times (\im i^{\aent}_{d-p-1}\cap \im i^{\coaent}_{d-p-1})$.
\end{claim}
\begin{proof}
	Suppose not. That is either
	\begin{enumerate}[label=(\roman*)]
		\item \label{case:1-Zel} there exists a $\hat\eta\in\coker i^\aent_p\cap\coker i^{\coaent}_p$ such that $\bbL(\hat\eta,\sigma)=0$, \\ for all $\sigma\in\im i^{\aent}_{d-p-1}\cap\im i^{\coaent}_{d-p-1}$ or
		\item \label{case:2-Zel} there exists a $\hat\sigma\in\im i^{\aent}_{d-p-1}\cap \im i^{\coaent}_{d-p-1}$ such that $\bbL(\eta,\hat\sigma)=0$, \\ for all $\eta\in\coker i^{\aent}_p\cap\coker i^{\coaent}_p$.
	\end{enumerate}
	If we suppose \labelcref{case:1-Zel}, then there is no homotopic obstruction to deforming $\hat\eta$ to either entirely $\aent$ or $\coaent$ which contradicts it lying in the cokernels of $i^\aent$ and $i^{\coaent}$. So let us suppose \labelcref{case:2-Zel} and pick a $\hat\sigma$ satisfying \labelcref{case:2-Zel}.  Because $\bbL$ is a non-degenerate pairing on $\H_p(\Sigma)\times \H_{d-p-1}(\Sigma)$ there exists a $\check\eta\in\H_p(\Sigma)$ such that $\bbL(\check\eta,\hat\sigma)\neq0$.  By assumption, $\check\eta$ must lie in either $\im i^\aent_p$ or $\im i^{\coaent}_p$ and so is completely deformable within $\Sigma$ to either $\aent$ or $\coaent$.  But then $\bbL(\check\eta,\hat\sigma)$ must actually vanish because $\hat\sigma$, which lies in $\im i^\aent_{d-p-1}\cap\im i^{\coaent}_{d-p-1}$, can be deformed to the respective complementary region so that it has no intersection with $\check\eta$. This contradiction completes the proof that $\bbL$ is non-degenerate on $(\coker i^{\aent}_p\cap \coker i^{\coaent}_p)\times (\im i^{\aent}_{d-p-1}\cap \im i^{\coaent}_{d-p-1})$.
\end{proof}

Similar arguments show that \(\h^{d-p-1}_\t{elec}=\h^p_\t{mag}\). As a consequence
\begin{equation}
	\abs{\fZ_\t{elec}}=\abs{\fZ_\t{mag}}.
\end{equation}

\section{Decomposition of the reduced density matrix}\label{app:redDMdecomp}

Given a state
\begin{equation}
	\rho=\sum_{\fv,\fv'}\rho_{\fv,\fv'}\,\kb{\fv}{\fv'},\qquad\rho_{\fv,\fv'}^\ast=\rho_{\fv',\fv},\qquad \sum_{\fv}\rho_{\fv,\fv}=1,
\end{equation}
we wish to write down the reduced density matrix, $\rho_\aent$, corresponding to the subregion algebra, $\fA_\t{mag}[\aent]$.  The general ansatz for this reduced density matrix is given by
\begin{equation}\label{eq:AmagredDMansatz1}
	\rho_\aent=\sum_{\set{\vec w^\aent_\sfi}}\sum_{\set{\vec v^\aent_\sfj}}C_{\set{\vec w^\aent_\sfi},\set{\vec v^\aent_\sfj}}\prod_{\eta^\sfi\hin\aent}\,\hat{\W}_{\eta^\sfi}^{\vec w^\aent_\sfi}\,\prod_{\sigma^\sfj\hin\aent}\hV_{\sigma^\sfj}^{\vec v^\aent_\sfj},
\end{equation}
for some coefficients $C$.  We notate the charges $\vec w^\aent$ and $\vec v^\aent$ to indicate that they are for surface operators deformable into $\aent$. Given that $\hW_{\eta^\sfi}^{\vec w_\sfi=\vec 0}$ and $\hV_{\sigma^\sfj}^{\vec v_\sfj=\vec 0}$ both act as the identity, it will be notationally useful to extend $\set{\vec w^\aent_\sfi}$ and $\set{\vec v^\aent_\sfj}$ to full charge vectors $\fw^\aent\in(\lattice_A)^\h$ and $\fv^\aent\in(\lattice_B)^\h$, respectively, with zero entries for all cycles not deformable into $\aent$:
\begin{equation}
	\vec w_\sfi^\aent=\vec v_\sfj^\aent=0\qquad\fall\,\eta^\sfi,\sigma^\sfj\nothin\aent.
\end{equation}
We then write
\begin{equation}\label{eq:AmagredDMansatz2}
	\rho_\aent=\sum_{\fw^\aent}\sum_{\fv^\aent}C_{\fw^\aent,\fv^\aent}\,\hW^{\fw^\aent}\,\hV^{\fv^\aent}.
\end{equation}
Hermiticity and unit-trace (with respect to $\mc H_\Sigma$) imply
\begin{equation}
	C^\ast_{\fw^\aent,\fv^\aent}=C_{-\fw^\aent,-\fv^\aent}\,\ex{\ii\Gamma\qty(\fv^{\aent},\fw^\aent)}
\end{equation}
and
\begin{equation}
	C_{\vec 0,\vec 0}=\qty(\text{dim}\mc H_\Sigma)^{-1}=\abs{\det\bbK}^{-\b_p(\Sigma)}\equiv\mc N_\Sigma^{-1},
\end{equation}
respectively.  We can solve for the coefficients $C_{\fw^\aent,\fv^\aent}$ in terms of the coefficients of the state, $\rho_{\fv,\fv'}$, by enforcing
\begin{equation}
	\Tr\left(\rho_\aent\,\mc O_\aent\right)=\Tr\left(\rho\,\mc O_\aent\right)=\sum_{\fv,\fv'}\rho_{\fv,\fv'}\mel{\fv'}{\mc O_\aent}{\fv},
\end{equation}
for all $\mc O_\aent\in\fA_\t{mag}[\aent]$. By considering a generic element $\mc O_\aent=\hW^{\hat\fw^{\aent}}\hV^{\hat\fv^{\aent}}$
(for fixed $\hat\fw^{\aent}$ and $\hat\fv^{\aent}$) we can easily work out
\begin{equation}
	C_{\fw^\aent,\fv^\aent}=\mc N_{\Sigma}^{-1}\sum_{\bar\fv}\rho_{\bar\fv,\bar\fv-\fv^\aent}\ex{-\ii\Gamma\qty(\bar\fv,\fw^\aent)}.
\end{equation}
Thus we can write a generic reduced density matrix as
\begin{equation}
	\rho_\aent=\mc N_{\Sigma}^{-1}\sum_{\fw^\aent}\sum_{\fv^\aent}\sum_{\bar\fv}\rho_{\bar\fv,\bar\fv-\fv^\aent}\ex{-\ii\Gamma\qty(\bar\fv,\fw^\aent)}\,\hW^{\fw^\aent}\,\hV^{\fv^\aent}.
\end{equation}
In particular for a pure state (as in \Cref{sect:redDM}),
\begin{equation}
	\rho=\sum_{\fv,\fv'}\psi_\fv\psi^\ast_{\fv'}\,\kb{\fv}{\fv'},
\end{equation}
the reduced density matrix is written as
\begin{equation}\label{eq:rdDMforpure}
	\rho_R=\mc N_{\Sigma}^{-1}\sum_{\fw^\aent}\sum_{\fv^\aent}\sum_{\bar\fv}\psi_{\bar\fv}\psi^\ast_{\bar\fv-\fv^\aent}\ex{-\ii\Gamma\qty(\bar\fv,\fw^\aent)}\,\hW^{\fw^\aent}\,\hV^{\fv^\aent}.
\end{equation}

\linespread{1.1}
\printbibliography

@article{Bonderson:2018ryx,
    author = "Bonderson, Parsa and Delaney, Colleen and Galindo, C\'esar and Rowell, Eric C. and Tran, Alan and Wang, Zhenghan",
    title = "{On invariants of Modular categories beyond modular data}",
    eprint = "1805.05736",
    archivePrefix = "arXiv",
    primaryClass = "math.QA",
    doi = "10.1016/j.jpaa.2018.12.017",
    journal = "J. Pure Appl. Algebra",
    volume = "223",
    pages = "4065--4088",
    year = "2019"
}

@article{kulkarni2018topological,
  title={A topological invariant for modular fusion categories},
  author={Kulkarni, Ajinkya and Mignard, Micha{\"e}l and Schauenburg, Peter},
  journal={arXiv preprint arXiv:1806.03158},
  year={2018}
}

@article{mignard2021modular,
  title={Modular categories are not determined by their modular data},
  author={Mignard, Micha{\"e}l and Schauenburg, Peter},
  journal={Letters in Mathematical Physics},
  volume={111},
  number={3},
  pages={60},
  year={2021},
  publisher={Springer},
  doi={10.1007/s11005-021-01395-0},
  eprint ={1708.02796},
  archivePrefix = "arXiv",
}

@article{Wen:2019ylt,
    author = "Wen, Xueda and Wen, Xiao-Gang",
    title = "{Distinguish modular categories and 2+1D topological orders beyond modular data: Mapping class group of higher genus manifold}",
    eprint = "1908.10381",
    archivePrefix = "arXiv",
    primaryClass = "cond-mat.str-el",
    month = "8",
    year = "2019"
}

@article{Qi:2012aa,
	author = {Xiao-Liang Qi},
	doi = {10.1103/PhysRevLett.108.196402},
	journal = {Physical Review Letters},
	number = {19},
	title = {General Relationship between the Entanglement Spectrum and the Edge State Spectrum of Topological Quantum States},
	url = {https://doi.org/10.1103/PhysRevLett.108.196402},
	volume = {108},
	year = {2012}}

@article{Buividovich:2008yv,
	archiveprefix = {arXiv},
	author = {Buividovich, P. V. and Polikarpov, M. I.},
	doi = {10.1088/1751-8113/42/30/304005},
	eprint = {0811.3824},
	journal = {PoS},
	pages = {039},
	primaryclass = {hep-lat},
	reportnumber = {ITEP-LAT-2008-21},
	title = {{Entanglement entropy in lattice gauge theories}},
	volume = {CONFINEMENT8},
	year = {2008}}

@article{Buividovich:2008kq,
	archiveprefix = {arXiv},
	author = {Buividovich, P. V. and Polikarpov, M. I.},
	doi = {10.1016/j.nuclphysb.2008.04.024},
	eprint = {0802.4247},
	journal = {Nucl. Phys. B},
	pages = {458--474},
	primaryclass = {hep-lat},
	reportnumber = {ITEP-LAT-2008-07},
	title = {{Numerical study of entanglement entropy in SU(2) lattice gauge theory}},
	volume = {802},
	year = {2008}}

@article{Cano:2014pya,
	archiveprefix = {arXiv},
	author = {Cano, Jennifer and Hughes, Taylor L. and Mulligan, Michael},
	doi = {10.1103/PhysRevB.92.075104},
	eprint = {1411.5369},
	journal = {Phys. Rev. B},
	number = {7},
	pages = {075104},
	primaryclass = {cond-mat.str-el},
	title = {{Interactions along an Entanglement Cut in 2+1D Abelian Topological Phases}},
	volume = {92},
	year = {2015}}

@article{chandran2011bulk,
	author = {Chandran, Anushya and Hermanns, M and Regnault, N and Bernevig, B Andrei},
	journal = {Physical Review B},
	number = {20},
	pages = {205136},
	publisher = {APS},
	title = {Bulk-edge correspondence in entanglement spectra},
	volume = {84},
	year = {2011}}

@article{Swingle:2011hu,
	archiveprefix = {arXiv},
	author = {Swingle, Brian and Senthil, T.},
	doi = {10.1103/PhysRevB.86.045117},
	eprint = {1109.1283},
	journal = {Phys. Rev. B},
	pages = {045117},
	primaryclass = {cond-mat.str-el},
	title = {{A Geometric proof of the equality between entanglement and edge spectra}},
	volume = {86},
	year = {2012}}

@article{Donnelly:2011hn,
	archiveprefix = {arXiv},
	author = {Donnelly, William},
	doi = {10.1103/PhysRevD.85.085004},
	eprint = {1109.0036},
	journal = {Phys. Rev. D},
	pages = {085004},
	primaryclass = {hep-th},
	title = {{Decomposition of entanglement entropy in lattice gauge theory}},
	volume = {85},
	year = {2012}}

@article{Li:2008kda,
	archiveprefix = {arXiv},
	author = {Li, Hui and Haldane, F.},
	doi = {10.1103/PhysRevLett.101.010504},
	eprint = {0805.0332},
	journal = {Phys. Rev. Lett.},
	number = {1},
	pages = {010504},
	primaryclass = {cond-mat.mes-hall},
	title = {{Entanglement Spectrum as a Generalization of Entanglement Entropy: Identification of Topological Order in Non-Abelian Fractional Quantum Hall Effect States}},
	volume = {101},
	year = {2008}}

@article{Apruzzi:2021nmk,
	archiveprefix = {arXiv},
	author = {Apruzzi, Fabio and Bonetti, Federico and Etxebarria, I\~naki Garc\'\i{}a and Hosseini, Saghar S. and Schafer-Nameki, Sakura},
	eprint = {2112.02092},
	month = {12},
	primaryclass = {hep-th},
	title = {{Symmetry TFTs from String Theory}},
	year = {2021}}

@article{Ye:2015eba,
	archiveprefix = {arXiv},
	author = {Ye, Peng and Gu, Zheng-Cheng},
	doi = {10.1103/PhysRevB.93.205157},
	eprint = {1508.05689},
	journal = {Phys. Rev. B},
	number = {20},
	pages = {205157},
	primaryclass = {cond-mat.str-el},
	title = {{Topological quantum field theory of three-dimensional bosonic Abelian-symmetry-protected topological phases}},
	volume = {93},
	year = {2016}}

@article{Tiwari:2016zru,
	archiveprefix = {arXiv},
	author = {Tiwari, Apoorv and Chen, Xiao and Ryu, Shinsei},
	doi = {10.1103/PhysRevB.95.245124},
	eprint = {1603.08429},
	journal = {Phys. Rev. B},
	number = {24},
	pages = {245124},
	primaryclass = {hep-th},
	title = {{Wilson operator algebras and ground states of coupled BF theories}},
	volume = {95},
	year = {2017}}

@article{Hsieh:2020jpj,
	archiveprefix = {arXiv},
	author = {Hsieh, Chang-Tse and Tachikawa, Yuji and Yonekura, Kazuya},
	doi = {10.1007/s00220-022-04333-w},
	eprint = {2003.11550},
	journal = {Commun. Math. Phys.},
	number = {2},
	pages = {495--608},
	primaryclass = {hep-th},
	reportnumber = {IPMU-20-0028, TU-1098},
	title = {{Anomaly Inflow and p-Form Gauge Theories}},
	volume = {391},
	year = {2022}}

@article{amabel2021differential,
	archiveprefix = {arXiv},
	author = {Araminta Amabel and Arun Debray and Peter J. Haine},
	eprint = {2109.12250},
	primaryclass = {math.AT},
	title = {Differential Cohomology: Categories, Characteristic Classes, and Connections},
	year = {2021}}

@article{Harlow:2020bee,
	archiveprefix = {arXiv},
	author = {Harlow, Daniel and Shaghoulian, Edgar},
	doi = {10.1007/JHEP04(2021)175},
	eprint = {2010.10539},
	journal = {JHEP},
	pages = {175},
	primaryclass = {hep-th},
	title = {{Global symmetry, Euclidean gravity, and the black hole information problem}},
	volume = {04},
	year = {2021}}

@article{Kallosh:1995hi,
    author = "Kallosh, Renata and Linde, Andrei D. and Linde, Dmitri A. and Susskind, Leonard",
    title = "{Gravity and global symmetries}",
    eprint = "hep-th/9502069",
    archivePrefix = "arXiv",
    reportNumber = "SU-ITP-95-2",
    doi = "10.1103/PhysRevD.52.912",
    journal = "Phys. Rev. D",
    volume = "52",
    pages = "912--935",
    year = "1995"
}

@article{Banks:2010zn,
    author = "Banks, Tom and Seiberg, Nathan",
    title = "{Symmetries and Strings in Field Theory and Gravity}",
    eprint = "1011.5120",
    archivePrefix = "arXiv",
    primaryClass = "hep-th",
    doi = "10.1103/PhysRevD.83.084019",
    journal = "Phys. Rev. D",
    volume = "83",
    pages = "084019",
    year = "2011"
}

@article{Harlow:2018tng,
	archiveprefix = {arXiv},
	author = {Harlow, Daniel and Ooguri, Hirosi},
	doi = {10.1007/s00220-021-04040-y},
	eprint = {1810.05338},
	journal = {Commun. Math. Phys.},
	number = {3},
	pages = {1669--1804},
	primaryclass = {hep-th},
	title = {{Symmetries in quantum field theory and quantum gravity}},
	volume = {383},
	year = {2021}}

@article{Freed:2022qnc,
	archiveprefix = {arXiv},
	author = {Freed, Daniel S. and Moore, Gregory W. and Teleman, Constantin},
	eprint = {2209.07471},
	month = {9},
	primaryclass = {hep-th},
	title = {{Topological symmetry in quantum field theory}},
	year = {2022}}

@article{Gaiotto:2020iye,
	archiveprefix = {arXiv},
	author = {Gaiotto, Davide and Kulp, Justin},
	doi = {10.1007/JHEP02(2021)132},
	eprint = {2008.05960},
	journal = {JHEP},
	pages = {132},
	primaryclass = {hep-th},
	title = {{Orbifold groupoids}},
	volume = {02},
	year = {2021}}

@article{Gaiotto:2014kfa,
	archiveprefix = {arXiv},
	author = {Gaiotto, Davide and Kapustin, Anton and Seiberg, Nathan and Willett, Brian},
	doi = {10.1007/JHEP02(2015)172},
	eprint = {1412.5148},
	journal = {JHEP},
	pages = {172},
	primaryclass = {hep-th},
	title = {{Generalized Global Symmetries}},
	volume = {02},
	year = {2015}}

@article{Bhardwaj:2017xup,
	archiveprefix = {arXiv},
	author = {Bhardwaj, Lakshya and Tachikawa, Yuji},
	doi = {10.1007/JHEP03(2018)189},
	eprint = {1704.02330},
	journal = {JHEP},
	pages = {189},
	primaryclass = {hep-th},
	reportnumber = {IPMU-17-0049},
	title = {{On finite symmetries and their gauging in two dimensions}},
	volume = {03},
	year = {2018}}

@inproceedings{Cordova:2022ruw,
	archiveprefix = {arXiv},
	author = {Cordova, Clay and Dumitrescu, Thomas T. and Intriligator, Kenneth and Shao, Shu-Heng},
	booktitle = {{Snowmass 2021}},
	eprint = {2205.09545},
	month = {5},
	primaryclass = {hep-th},
	title = {{Snowmass White Paper: Generalized Symmetries in Quantum Field Theory and Beyond}},
	year = {2022}}

@article{Jiang:2012uea,
	archiveprefix = {arXiv},
	author = {Jiang, Hong Chen and Wang, Zhenghan and Balents, Leon},
	doi = {10.1038/nphys2465},
	eprint = {1205.4289},
	journal = {Nature Phys.},
	number = {12},
	pages = {902--905},
	primaryclass = {cond-mat.str-el},
	title = {{Identifying topological order by entanglement entropy}},
	volume = {8},
	year = {2012}}

@article{VanAcoleyen:2015ccp,
	archiveprefix = {arXiv},
	author = {Van Acoleyen, Karel and Bultinck, Nick and Haegeman, Jutho and Marien, Michael and Scholz, Volkher B. and Verstraete, Frank},
	doi = {10.1103/PhysRevLett.117.131602},
	eprint = {1511.04369},
	journal = {Phys. Rev. Lett.},
	number = {13},
	pages = {131602},
	primaryclass = {quant-ph},
	title = {{The entanglement of distillation for gauge theories}},
	volume = {117},
	year = {2016}}

@article{Slagle:2020ugk,
	archiveprefix = {arXiv},
	author = {Slagle, Kevin},
	doi = {10.1103/PhysRevLett.126.101603},
	eprint = {2008.03852},
	journal = {Phys. Rev. Lett.},
	number = {10},
	pages = {101603},
	primaryclass = {hep-th},
	title = {{Foliated Quantum Field Theory of Fracton Order}},
	volume = {126},
	year = {2021}}

@article{Wen:2016snr,
	archiveprefix = {arXiv},
	author = {Wen, Xueda and Matsuura, Shunji and Ryu, Shinsei},
	doi = {10.1103/PhysRevB.93.245140},
	eprint = {1603.08534},
	journal = {Phys. Rev. B},
	number = {24},
	pages = {245140},
	primaryclass = {cond-mat.mes-hall},
	title = {{Edge theory approach to topological entanglement entropy, mutual information and entanglement negativity in Chern-Simons theories}},
	volume = {93},
	year = {2016}}

@article{Fliss:2017wop,
	archiveprefix = {arXiv},
	author = {Fliss, Jackson R. and Wen, Xueda and Parrikar, Onkar and Hsieh, Chang-Tse and Han, Bo and Hughes, Taylor L. and Leigh, Robert G.},
	doi = {10.1007/JHEP09(2017)056},
	eprint = {1705.09611},
	journal = {JHEP},
	pages = {056},
	primaryclass = {cond-mat.str-el},
	reportnumber = {JHEP-09-(2017)-056},
	title = {{Interface Contributions to Topological Entanglement in Abelian Chern-Simons Theory}},
	volume = {09},
	year = {2017}}

@article{Fliss:2020cos,
	archiveprefix = {arXiv},
	author = {Fliss, Jackson R. and Leigh, Robert G.},
	doi = {10.1007/JHEP07(2020)009},
	eprint = {2004.05123},
	journal = {JHEP},
	pages = {009},
	primaryclass = {hep-th},
	title = {{Interfaces and the extended Hilbert space of Chern-Simons theory}},
	volume = {07},
	year = {2020}}

@inproceedings{Faulkner:2022mlp,
	archiveprefix = {arXiv},
	author = {Faulkner, Thomas and Hartman, Thomas and Headrick, Matthew and Rangamani, Mukund and Swingle, Brian},
	booktitle = {{Snowmass 2021}},
	eprint = {2203.07117},
	month = {3},
	primaryclass = {hep-th},
	reportnumber = {BRX-TH-6703},
	title = {{Snowmass white paper: Quantum information in quantum field theory and quantum gravity}},
	year = {2022}}

@article{Soni:2015yga,
	archiveprefix = {arXiv},
	author = {Soni, Ronak M and Trivedi, Sandip P.},
	doi = {10.1007/JHEP01(2016)136},
	eprint = {1510.07455},
	journal = {JHEP},
	pages = {136},
	primaryclass = {hep-th},
	reportnumber = {TIFR-TH-15-29},
	title = {{Aspects of Entanglement Entropy for Gauge Theories}},
	volume = {01},
	year = {2016}}

@article{Donnelly:2014fua,
	archiveprefix = {arXiv},
	author = {Donnelly, William and Wall, Aron C.},
	doi = {10.1103/PhysRevLett.114.111603},
	eprint = {1412.1895},
	journal = {Phys. Rev. Lett.},
	number = {11},
	pages = {111603},
	primaryclass = {hep-th},
	title = {{Entanglement entropy of electromagnetic edge modes}},
	volume = {114},
	year = {2015}}

@article{Buividovich:2008gq,
	archiveprefix = {arXiv},
	author = {Buividovich, P. V. and Polikarpov, M. I.},
	doi = {10.1016/j.physletb.2008.10.032},
	eprint = {0806.3376},
	journal = {Phys. Lett. B},
	pages = {141--145},
	primaryclass = {hep-th},
	reportnumber = {ITEP-LAT-2008-14},
	title = {{Entanglement entropy in gauge theories and the holographic principle for electric strings}},
	volume = {670},
	year = {2008}}

@article{chamon2005quantum,
	archiveprefix = {arXiv},
	author = {Chamon, Claudio},
	doi = {10.1103/PhysRevLett.94.040402},
	eprint = {cond-mat/0404182},
	issue = {4},
	journal = {Phys. Rev. Lett.},
	month = {1},
	numpages = {4},
	pages = {040402},
	publisher = {American Physical Society},
	title = {Quantum Glassiness in Strongly Correlated Clean Systems: An Example of Topological Overprotection},
	volume = {94},
	year = {2005}}

@article{bravyi2011topological,
	archiveprefix = {arXiv},
	author = {Sergey Bravyi and Bernhard Leemhuis and Barbara M. Terhal},
	doi = {10.1016/j.aop.2010.11.002},
	eprint = {1006.4871},
	issn = {0003-4916},
	journal = {Annals of Physics},
	number = {4},
	pages = {839-866},
	primaryclass = {quant-ph},
	title = {Topological order in an exactly solvable 3D spin model},
	volume = {326},
	year = {2011}}

@article{bravyi2013quantum,
	archiveprefix = {arXiv},
	author = {Bravyi, Sergey and Haah, Jeongwan},
	doi = {10.1103/PhysRevLett.111.200501},
	eprint = {1112.3252},
	issue = {20},
	journal = {Phys. Rev. Lett.},
	month = {11},
	numpages = {5},
	pages = {200501},
	primaryclass = {quant-ph},
	publisher = {American Physical Society},
	title = {Quantum Self-Correction in the 3D Cubic Code Model},
	volume = {111},
	year = {2013}}

@article{Zou:2016dck,
	archiveprefix = {arXiv},
	author = {Zou, Liujun and Haah, Jeongwan},
	doi = {10.1103/PhysRevB.94.075151},
	eprint = {1604.06101},
	journal = {Phys. Rev. B},
	number = {7},
	pages = {075151},
	primaryclass = {cond-mat.str-el},
	title = {{Spurious long-range entanglement and replica correlation length}},
	volume = {94},
	year = {2016}}

@article{Kim:2023ydi,
	archiveprefix = {arXiv},
	author = {Kim, Isaac H. and Levin, Michael and Lin, Ting-Chun and Ranard, Daniel and Shi, Bowen},
	eprint = {2302.00689},
	month = {2},
	primaryclass = {quant-ph},
	title = {{Universal lower bound on topological entanglement entropy}},
	year = {2023}}

@article{Magan:2020ake,
	archiveprefix = {arXiv},
	author = {Magan, Javier M. and Pontello, Diego},
	doi = {10.1103/PhysRevA.103.012211},
	eprint = {2005.01760},
	journal = {Phys. Rev. A},
	number = {1},
	pages = {012211},
	primaryclass = {hep-th},
	title = {{Quantum Complementarity through Entropic Certainty Principles}},
	volume = {103},
	year = {2021}}

@unpublished{AkersSoni,
	author = {Akers, Chris and Soni, Ronak and Wei, Annie},
	title = {{to appear}}}

@article{Fliss:2023uiv,
    author = "Fliss, Jackson R. and Vitouladitis, Stathis",
    title = "{Entanglement in BF theory II: Edge-modes}",
    eprint = "2310.18391",
    archivePrefix = "arXiv",
    primaryClass = "hep-th",
    month = "10",
    year = "2023"
}

@article{Radicevic:2015sza,
	archiveprefix = {arXiv},
	author = {Radi\v{c}evi\'c, \DH{}or\dj{}e},
	doi = {10.1007/JHEP04(2016)163},
	eprint = {1509.08478},
	journal = {JHEP},
	pages = {163},
	primaryclass = {hep-th},
	reportnumber = {SU-ITP-15-13},
	title = {{Entanglement in Weakly Coupled Lattice Gauge Theories}},
	volume = {04},
	year = {2016}}

@article{Lin:2018bud,
	archiveprefix = {arXiv},
	author = {Lin, Jennifer and Radi\v{c}evi\'c, \DJ{}or\dj{}e},
	doi = {10.1016/j.nuclphysb.2020.115118},
	eprint = {1808.05939},
	journal = {Nucl. Phys. B},
	pages = {115118},
	primaryclass = {hep-th},
	title = {{Comments on defining entanglement entropy}},
	volume = {958},
	year = {2020}}

@article{Levin:2006zz,
	archiveprefix = {arXiv},
	author = {Levin, Michael and Wen, Xiao-Gang},
	doi = {10.1103/PhysRevLett.96.110405},
	eprint = {cond-mat/0510613},
	journal = {Phys. Rev. Lett.},
	pages = {110405},
	title = {{Detecting Topological Order in a Ground State Wave Function}},
	volume = {96},
	year = {2006}}

@article{Balasubramanian:2016sro,
	archiveprefix = {arXiv},
	author = {Balasubramanian, Vijay and Fliss, Jackson R. and Leigh, Robert G. and Parrikar, Onkar},
	doi = {10.1007/JHEP04(2017)061},
	eprint = {1611.05460},
	journal = {JHEP},
	pages = {061},
	primaryclass = {hep-th},
	title = {{Multi-Boundary Entanglement in Chern-Simons Theory and Link Invariants}},
	volume = {04},
	year = {2017}}

@article{Bergeron:1994ym,
	archiveprefix = {arXiv},
	author = {Bergeron, Mario and Semenoff, Gordon W. and Szabo, Richard J.},
	doi = {10.1016/0550-3213(94)00503-7},
	eprint = {hep-th/9407020},
	journal = {Nucl. Phys. B},
	pages = {695--722},
	reportnumber = {MIT-CTP-2326, UBCTP-94-004},
	title = {{Canonical bf type topological field theory and fractional statistics of strings}},
	volume = {437},
	year = {1995}}

@article{Delcamp:2016eya,
	archiveprefix = {arXiv},
	author = {Delcamp, Clement and Dittrich, Bianca and Riello, Aldo},
	doi = {10.1007/JHEP11(2016)102},
	eprint = {1609.04806},
	journal = {JHEP},
	pages = {102},
	primaryclass = {hep-th},
	title = {{On entanglement entropy in non-Abelian lattice gauge theory and 3D quantum gravity}},
	volume = {11},
	year = {2016}}

@article{Jian:2015wra,
	archiveprefix = {arXiv},
	author = {Jian, Chao-Ming and Kim, Isaac H. and Qi, Xiao-Liang},
	eprint = {1508.07006},
	month = {8},
	primaryclass = {cond-mat.str-el},
	title = {{Long-range mutual information and topological uncertainty principle}},
	year = {2015}}

@article{Blau:1989bq,
	author = {Blau, Matthias and Thompson, George},
	doi = {10.1016/0003-4916(91)90240-9},
	journal = {Annals Phys.},
	pages = {130--172},
	reportnumber = {SISSA-39/89/FM, PAR-LPTHE-89-17},
	title = {{Topological Gauge Theories of Antisymmetric Tensor Fields}},
	volume = {205},
	year = {1991}}

@article{Gegenberg:1993gd,
	archiveprefix = {arXiv},
	author = {Gegenberg, J. and Kunstatter, G.},
	doi = {10.1006/aphy.1994.1043},
	eprint = {hep-th/9304016},
	journal = {Annals Phys.},
	pages = {270--289},
	reportnumber = {WIN-92-03-REV, UNB-92-01-REV},
	title = {{The Partition function for topological field theories}},
	volume = {231},
	year = {1994}}

@book{hatcher,
	author = {Hatcher, Allen},
	isbn = {0-521-79540-0},
	publisher = {Cambridge University Press},
	title = {Algebraic topology},
	url = {https://pi.math.cornell.edu/~hatcher/AT/ATpage.html},
	year = 2002}

@article{kitaev2006topological,
	archiveprefix = {arXiv},
	author = {Kitaev, Alexei and Preskill, John},
	doi = {10.1103/PhysRevLett.96.110404},
	eprint = {hep-th/0510092},
	journal = {Phys. Rev. Lett.},
	pages = {110404},
	reportnumber = {CALT-68-2578},
	title = {{Topological entanglement entropy}},
	volume = {96},
	year = {2006}}

@article{Grover:2011fa,
	archiveprefix = {arXiv},
	author = {Grover, Tarun and Turner, Ari M. and Vishwanath, Ashvin},
	doi = {10.1103/PhysRevB.84.195120},
	eprint = {1108.4038},
	journal = {Phys. Rev. B},
	pages = {195120},
	primaryclass = {cond-mat.str-el},
	title = {{Entanglement Entropy of Gapped Phases and Topological Order in Three dimensions}},
	volume = {84},
	year = {2011}}

@article{Casini:2019kex,
	archiveprefix = {arXiv},
	author = {Casini, Horacio and Huerta, Marina and Mag\'an, Javier M. and Pontello, Diego},
	doi = {10.1007/JHEP02(2020)014},
	eprint = {1905.10487},
	journal = {JHEP},
	pages = {014},
	primaryclass = {hep-th},
	title = {{Entanglement entropy and superselection sectors. Part I. Global symmetries}},
	volume = {02},
	year = {2020}}

@article{blauMassiveRaySingerTorsion2022,
	archiveprefix = {arXiv},
	author = {Blau, Matthias and Kakona, Mbambu and Thompson, George},
	eprint = {2206.12268},
	eprinttype = {arxiv},
	month = jun,
	number = {arXiv:2206.12268},
	primaryclass = {hep-th},
	publisher = {{arXiv}},
	title = {Massive {{Ray-Singer Torsion}} and {{Path Integrals}}},
	year = {2022}}

@article{Casini:2013rba,
	archiveprefix = {arXiv},
	author = {Casini, Horacio and Huerta, Marina and Rosabal, Jose Alejandro},
	doi = {10.1103/PhysRevD.89.085012},
	eprint = {1312.1183},
	journal = {Phys. Rev. D},
	number = {8},
	pages = {085012},
	primaryclass = {hep-th},
	title = {{Remarks on entanglement entropy for gauge fields}},
	volume = {89},
	year = {2014}}

@article{Hofman:2024oze,
    author = "Hofman, Diego M. and Vitouladitis, Stathis",
    title = "{Generalised symmetries and state-operator correspondence for nonlocal operators}",
    eprint = "2406.02662",
    archivePrefix = "arXiv",
    primaryClass = "hep-th",
    month = "6",
    year = "2024"
}
\end{document}